\documentclass{amsart}
\usepackage[english]{babel}
\usepackage{amsfonts}
\usepackage{amssymb}
\usepackage{color}

\newtheorem{thm}{Theorem}[section]

\newtheorem{asm}[thm]{Assumption}
\newtheorem{lem}[thm]{Lemma}

\newtheorem{prop}[thm]{Proposition}
\theoremstyle{definition}

\newtheorem{exmp}{Example}[section]
\theoremstyle{remark}
\newtheorem{rem}{Remark}[section]

\numberwithin{equation}{section}

\renewcommand{\i}[1]{\raisebox{1pt}{$\stackrel{\scriptscriptstyle \circ}{#1}$}}

\begin{document}

\title[Stokes type operators]{Maxwell's and Stokes' 
operators associated with elliptic differential 
complexes}

\author{A.A. Shlapunov}
\address[Alexander Shlapunov]
{Siberian Federal University
                                                 \\
         pr. Svobodnyi 79
                                                 \\
         660041 Krasnoyarsk
                                                 \\
         Russia}
\email{ashlapunov@sfu-kras.ru}  

\author{A.N. Polkovnikov}
\address[Alexander Polkovnikov]
{Siberian Federal University
                                                 \\
         pr. Svobodnyi 79
                                                 \\
         660041 Krasnoyarsk
                                                 \\
         Russia}
\email{paskaattt@yandex.ru} 

\author{V.L. Mironov}
\address[Victor Mironov]
{Institute for Physics of Microstructures \\ Russian Academy of Sciences\\  GSP-105, Nizhny Novgorod 603950 \\
Russia} 
\email{mironov@ipmras.ru}


\begin{abstract}
We propose a new technique to generate reasonable systems of partial differential equations (PDE)
that could be potential candidates for depicting models in natural sciences related to quasi-linear 
 equations. Such systems appear within typical constructions of the Homological Algebra 
as complexes of differential operators describing compatibility conditions for overdetermined 
systems of PDE's. The related models can be both steady and evolutionary. Additional assumptions on 
the ellipticity of the differential complex provide a wide class of elliptic, parabolic and hyperbolic 
operators that could be generated in this way. In particular, it appears that an essentially large 
amount of equations related to the modern Mathematical Physics is generated by the de Rham complex of 
differentials on the exterior differential forms. These includes the elliptic Laplace and Lam\'e type 
operators; the parabolic heat transfer equation; the Euler type and Navier-Stokes type equations in 
Hydrodynamics; the hyperbolic wave equation and  the Maxwell equations in Electrodynamics; 
the Klein-Gordon equation in Relativistic Quantum Mechanics; and so on. Our model generation method covers a 
broad class of generating systems, especially in higher spatial dimensions, due to different basic algebraic 
structures at play. 
\end{abstract}

\maketitle

\section*{Introduction}
\label{s.Int}

The vast majority of  differential equations of modern Mathematical Physics was constructed 
with the use of the standard time derivative $\partial _t=\partial /\partial t$, gradient 
operator $\nabla$, the divergence operator $\mathrm{div}$  and the infinitesimal circulation 
operator  $\mathrm{curl}$, known since  Hamilton \cite{Ham} and  Maxwell \cite{Maxw}. 
The operators satisfy familiar relations 
\begin{equation} \label{eq.deRham.base.rel.3}
\mathrm{curl} \circ \nabla =0, \, \mathrm{div} \circ \mathrm{curl} =0, 
\end{equation}
generating the elliptic Laplace operator 
$$
\Delta=\mathrm{div} \circ \nabla, 
$$
which is used in the parabolic heat transfer and duffusion equations (operator $\partial _t - \Delta$) and in the 
hyperbolic wave equations (operator $\partial^2 _t - \Delta$). 

Advanced algebraic concepts, such as Dirac matrix algebra, 
Pauli matrix algebra, Clifford  algebra, quaternionic (octonionic, sedenionic)  
constructions and so on,  were used within Mathematical Physics in order to express 
physical laws in more clear and compact ways, see, for instance, \cite{Dirac}, \cite{Pauli}, 
\cite{Mir18}, \cite{Mjr} 
and many others.

In this paper, using Stokes' system of Hydrodynamics as a model example, we propose a more 
general algebraic construction related to Homological Algebra that helps to describe physical 
laws in a unified form with the use of elliptic differential complexes, see, instance, \cite{Tark35}. 
The approach does not give a precise description of the related models, but it 
suggest dimensions of the corresponding known and unknown vectors and type of equations up to 
(both linear and non-linear) perturbations. 
The other details depend usually on the particular type of the processes and symmetries behind them. 

Of course, there are other ways to generate mathematical models in standardised ways, 
for instance, in the frame of  General Relativity Theory, see, for example, \cite{Ren}. But 
in the present paper, instead explaining how considerations in Physics involve partial 
differential equations, we illustrate how a system of PDE's may generate more extensive 
mathematical model within Mathematical Physics.

We also indicate simple conditions, providing (Petrovskii or Douglis-Nirenberg) ellipticity 
of the related steady Maxwell's and Stokes' type systems and, consequently, parabolicity or 
hyperbolicity of the related time dependent systems. 
Actually, this opens a way to construct easily parametrices and fundamental 
solutions to Maxwell' and Stokes' type operators under the considerations. 

\section{Differential complexes}
\label{s.DC}

Let us shortly recall the notion of differential complex and related matters. 

\subsection{Differential operators} 

Let $X$ be a $C^{\infty}$-smooth Riemannian manifold of dimension $n\geq 2$ with a smooth (possibly, empty) 
boundary $\partial X$. We tacitly assume that it is enclosed into a smooth 
manifold $\tilde{X}$ of the same dimension. Let also $\i{X}$ denotes the interior 
of $X$.

For any smooth $\mathbb C$-vector bundles $E$ and $F$ of rangs $k$ and $l$, respectively, over $X$, we write 
$\mathrm{Diff}_{m} (X; E \to F)$ for the space of all the linear partial 
differential operators of order $\leq m\in {\mathbb Z}_+$ between sections of the bundles $E$ and $F$. 
Then, for an open set ${\mathcal O}\subset \i{X}$ 
over which the bundles and the manifold are trivial, the sections over ${\mathcal O}$ may 
be interpreted as (vector-) functions and $A\in \mathrm{Diff}_m(X; E\rightarrow 
F)$ is given as $(l\times k)$-matrix of scalar differential operators, i.e. we 
have 
\begin{equation*} 
A=A(x,D)=\sum_{|\alpha|\leq m}a_\alpha(x) \partial^\alpha
\end{equation*}
where $a_\alpha(x)$ are $(l\times k)$-matrices of $C^\infty ({\mathcal O})$-functions, 
$\partial_j = \frac{\partial}{\partial x_j}$,  $\partial^\alpha =
 \partial_1^{\alpha_1} \dots \partial_1^{\alpha_n}$. Denote by $I_{k}$ the identity operator on sections 
of the bundle $E$ (a unit $(k \times k)$-matrix in the local situation) and by  $E^{\ast}$ the conjugate 
bundle of $E$. Any Hermitian metric $(.,.)_{E,x}$ on $E$ gives rise to a sesquilinear bundle
isomorphism (the Hodge operator) $\star_{E} \! : E \to E^{\ast}$ by the 
equality $\langle \star_{E} v, u \rangle_{E,x} = (u,v)_{E,x}$ for all sections $u$ 
and $v$ of $E$; here $\langle ., . \rangle_{E,x}$ is the natural pairing in the 
fibers of $E^*$ and $E$. 
Pick a volume form $dx$ on $X$, thus identifying the dual bundle, the conjugate bundle and 
the Lebesgue space $L^2 (E)$ with the inner product induced by $(.,.)_{E,x}$. 
Then for $A \in \mathrm{Diff}_{m} (X; E \to F)$ denote  
by $A^{\ast} \in \mathrm{Diff}_{m} (X; F \to E)$ the corresponding formal adjoint operator.
Let also $\mathcal D$ be a bounded domain (i.e. open connected set) in $\i{X}$. 

\subsection{Compatibility differential complexes}

We recall that  a differential operator $A$ is called overdetermined on $X$ if there is a non-zero differential 
operator $B$ over $X$ such that
\begin{equation} \label{eq.compat}
B \circ A \equiv 0.
\end{equation}

An operator $B$, satisfying \eqref{eq.compat}, is called a compatibility operator for $A$ 
if for any operator $\tilde B$ satisfying $\tilde B \circ A \equiv 0$ there is an operator $C$
such that $\tilde B = C \circ B$. Clearly, a compatibility operator is not uniquely defined;
however it gives necessary solvability conditions to the operator equation  
$$
A u =f \mbox{ in } \mathcal D
$$
in a domain $\mathcal{D}\subset X$, i.e. $Bf=0$. 
 However, a compatibility operator may also contain addition information on a physical 
model where the operator $A$ 
appeared. Of course, the operator $B$ can also be overdetermined.

Thus, our principal object to discuss will be a complex $\{A_q,E_q\} _{q=0}^N$ of partial 
differential operators over $X$ (see, for instance, \cite{Spe}, \cite{Tark35}), 
\begin{equation} \label{eq.complex}
0\rightarrow 
C^\infty(E_0)\stackrel{A_0}{\rightarrow}C^\infty(E_1)\stackrel{A_1}{
\rightarrow}C^\infty(E_2) \rightarrow \dots 
\stackrel{A_{N-1}}{\rightarrow} C^\infty(E_N) \rightarrow 0, 
\end{equation}
where $E_q$ are  bundles of rangs $k_q$, respectively,  over $X$ and $A_q$ are differential 
operators from $\mathrm{Diff}_{m_q} (X; E_q \to 
E_{q+1})$ with 
\begin{equation} \label{eq.complex.base.rel}
A_{q+1} \circ A_q \equiv 0;
\end{equation} 
we  assume that $A_q =0$ for both $q<0$ and $q\geq N$. Actually, it is often convenient to consider 
complex $\{A_q,E_q\}_{q=0}^N$ as a graduated operator $A^\cdot$ of degree $1$ over a graduated 
topological vector space ${\mathfrak S}^{\cdot} =\oplus_{q=0}^N {\mathfrak S}^q (E_q)$  in such a way that 
$A^\cdot u =A_q u $ for a section $u\in {\mathfrak S}^q (E_q)$ of the bundle $E_q$.  

It may happens, see examples below, that orders $m_q$ of the differential operators $A_q$ are 
different; so we set $m=\max\limits_{0\leq q \leq N-1} m_q$. But the most simple constructions 
corresponds to the cases where
\begin{equation} \label{eq.all.equal}
m_j=m \mbox{ for all } 0\leq j\leq N-1. 
\end{equation}
As above, we fix Hermitian metrics 
$(\cdot,\cdot)_{q,x} =(\cdot,\cdot)_{E_q,x}$ in each fiber $E_{q,x}$. 

The differential complex $\{A_q,E_q\}$ is called a compatibility complex for $A_0$ if for each $q\geq 0$ 
the differential operator $A_{q+1}$ is a compatibility operator for $A_q$. 
As the compatibility operator is not unique, the compatibility 
complex is not unique, too. The notions of homotopical equivalence of complexes (\cite[Definition 1.1.17]{Tark35}) 
and equivalent operators \cite[Definition 1.2.5]{Tark35}) help to improve the situation. 
In particular, homotopically equivalent complexes have isomorphic cohomologies over 
many standard functional classes, see \cite[Proposition 1.24]{Tark35}. 
According to \cite[Propositions 1.2.7 and 1.2.8]{Tark35}, if differential 
operators $A_0$ and $\tilde A_0$ are equivalent  and the operator $A_0$ is included into 
a compatibility complex $\{A_q,E_q\}_{i=0}^N$ then for the operator $\tilde A_0$
there is a compatibility complex $\{\tilde A_q,\tilde E_q\}_{q=0}^N$ and, moreover, 
the corresponding complexes are homotopicaly equivalent.

The algebraic structures lying at the bottom of the theory of differential complexes are
rather natural, see \cite[Ch. 1]{Tark35}, though this depends on the class of considered 
operators. Namely, let us consider the two typical cases.
 
If $X={\mathbb R}^n$ and $A=A(D)$ is an $(l\times k)$-matrix differential operator with constant 
coefficients then one may use $\mathcal P$-modules of the  ring $\mathcal P$ of all the polynomials 
with complex coefficients, see \cite{Pal}, \cite[\S  1.2]{Tark35}, or elsewhere. Let us denote by 
${\mathcal P}^k$ the direct sum of $k$ copies of the ring $\mathcal P$ and denote 
by $A(\zeta)$ the polynomial matrix 
\begin{equation*} 
A (\zeta)=\sum_{|\alpha|\leq m}a_\alpha (\iota \zeta )^\alpha, \,\, \zeta \in {\mathbb C}^n.
\end{equation*}
 Then the transposed matrix $A'(\zeta)$ naturally defines a mapping $A'(\zeta): {\mathcal P}^l \to 
{\mathcal P}^k $. As the ring $\mathcal P$  is Noetherian, then the $\mathcal P$-module 
${\mathcal P}^k / A'(\zeta) {\mathcal P}^l $ is finitely generated, i.e. there are natural numbers $N$, 
$k_1, \dots k_N$ and polynomial $(k_{i+1} \times k_i)$-matrices 
$$
A_q (\zeta) =  \sum_{|\alpha|\leq m_q} a^{(q)}_\alpha (\iota \zeta)^\alpha, \,\, 0\leq q \leq N-1,
$$
such that $k_0=k$, $k_1=l$, $A _0(\zeta)= A (\zeta)$, 
$ A'_{q} (\zeta) \circ A'_{q+1} (\zeta) =0$  for all $ \zeta \in \mathbb C$, 
and the following sequence is exact, see Hilbert Syzygies Theorem, \cite[\S 8]{Bou}, 
\begin{equation*} 
0\leftarrow {\mathcal P}^k / A'(\zeta) {\mathcal P}^l  \leftarrow {\mathcal P}^{k_0} 
\stackrel{A '_0(z)}{\longleftarrow} {\mathcal P}^{k_1}  
\stackrel{ A' _1(z)} {\longleftarrow} {\mathcal P}^{k_2}  \leftarrow  \dots  
\stackrel{ A_{N-1}'(z)} {\longleftarrow} {\mathcal P}^{k_N} \leftarrow 0.
\end{equation*} 

Then the related operators 
$$
A_q (D) =\sum_{|\alpha|\leq m_q} a^{(q)}_\alpha \partial^\alpha
$$ 
with constant coefficients form the desired compatibility differential complex \eqref{eq.complex} for 
$A=A_0$; it is called the Hilbert complex associated with the ${\mathcal P}$-module for $A$. 
Actually, if $ \mathcal D$ is a convex domain in ${\mathbb R}^n$ then for any section $f \in C^\infty 
(\mathcal{D},E_{q+1})$ satisfying $A_{q+1} f=0$  in $\mathcal D$  there is $u \in C^\infty 
(\mathcal{D},E_{q})$ satisfying $A  _{q} u=f $ in $\mathcal D$, i.e. the Hilbert complex gives both 
necessary and  sufficient conditions for the solvability of 
the related operator equations in this particular situation, see, for instance, \cite{Pal}. 

In the general case of differential operators with variable coefficients (or even operators on 
manifold), to construct a compatibility complex for an operator $A$ is a more delicate procedure, 
see \cite{Finik}, \cite{Riq}, \cite{Spe}. In particular, the related complex might be not finite. 
D.C. Spencer \cite{Spe} granted existence of a  (finite) compatibility differential complex for any 
"sufficiently regular"{} differential operator with infinitely smooth coefficients, see also 
\cite[\S 1.3]{Tark35} for a more advanced discussion. To define the concept "sufficient regularity"{} 
one should consider jets $j^s$ of sections $E$ and $F$ over $X$ of finite length $s$ and the 
prolongations $j^s \circ A$ of the differential operator $A$ to the spaces of jets 
${\mathcal J}^s (E)$, see \cite[\S 1.3.2]{Tark35}. More precisely, let  $\eta (A) : J^{m} (E) \to F$ 
be the bundle homomorphism satisfying $\eta(A) \circ j^m =A$ and   
$$
{\mathfrak R}^s (x) = \ker \{\eta (j^{s-m} \circ A): {\mathcal J}^s  (E)_x \to 
{\mathcal J}^{s-m}  (F)_x  \}, \, \, x \in X.
$$
The operator $A$ is called "sufficiently regular" if 1) the dimensions $d(s,x)$ of the spaces ${\mathfrak R}^s (x)$ 
do not depend on $x \in X$ for $s\geq m$ and 2) the natural "projections" $\pi^{s_2,s_1}: {\mathfrak R}^{s_2} (x) 
\to {\mathfrak R}^{s_1} (x)$ have constant rank for all $s_2\geq s_1 \geq m$. 
Of course, the operators with constant coefficients are "sufficiently regular". 

\subsection{Elliptic differential complexes}
 
Let $\pi: T^{\ast} X \to X$ be the (real) cotangent bundle of $X$ and 
let $\pi^*E $ be a induced bundle for the bundle $E$ (i.e. the fiber 
of $\pi^*E$ over the point $(x,\zeta) \in T^{\ast} X $ coincides with $E_x$). 
We write $\sigma  (A): \pi^*E  \to \pi^*F $ for the principal homogeneous 
symbol of the order $m$ of the operator $A$, see, for instance, \cite[\S 1.1.9]{Tark35}. 
Of course, in a suitable local chart we have 
$$
\sigma(A) (x,\zeta) =\sum_{|\alpha|= m} a_\alpha(x) (\iota\zeta)^\alpha, \, x\in {\mathcal O}, \, 
\zeta \in {\mathbb R}^n,
$$ 
where $\iota$ is the imaginary unit. 
We recall that $A$ is called elliptic on $X$ if $k=l$ and the mapping 
$\sigma  (A) (x,\zeta): \pi^*E_x  \to \pi^*F_x $ is invertible for 
$(x,\zeta) \in T^{\ast} X $ with $\zeta\ne 0$, see, for instance, 
\cite[Ch 1, \S 3, Ch. 2, \S 2]{EgShu}. Sometimes $A$ is called {\it overdetermined elliptic} if 
$k< l$ and the mapping $\sigma  (A) (x,\zeta): \pi^*E \to \pi^*F $ is injective for all 
$(x,\zeta) \in T^{\ast} X $ with $\zeta\ne 0$, but not surjective for some  $(x,\zeta)$. 
A typical operator with injective symbol is a suitable connection $\nabla_{E}$ related to a bundle 
$E$, i.e. a first differential operator of the type 
$E \to E\otimes T^*X$, compatible with Hermitian metric 
$(\cdot,\cdot)_{E,x}$ in each fiber $E_x$, see, for instance, \cite[Ch. III]{Wells}. 
In particular, for a trivial vector bundle $E= {\mathbb R}^n\times 
{\mathbb C}^k$ we have $\nabla _E = I_k  \otimes \nabla$ with the usual gradient operator 
$\nabla $ in ${\mathbb R}^n$ where $M_1 \otimes M_2$ stands for the tensor product of matrices $M_1$ and $M_2$. 

Recall that an operator $A$ of an even order $m=2p$ and of type $E\to E$ 
is called strongly elliptic, if  
$$
\Re ( \sigma(A) (x,\zeta) \,w, w)_{E,x}  >0 
 \mbox{ for all } (x,\zeta) \in T^{\ast} X \setminus \{0\} , w \in E_x \setminus \{0\}, 
$$
where $\Re \, a$ is the real part of a complex number $a$. A typical strongly elliptic operator of the second order 
is given by the 'Laplacian' $\nabla_E ^* \nabla _E$. For a trivial vector bundle $E= 
{\mathbb R}^n\times {\mathbb C}^k$ we have $\nabla_E ^* \nabla _E = - I_k\otimes \Delta$ with the 
usual Laplace operator $\Delta $ in ${\mathbb R}^n$. 

A more general notion of ellipticity was introduced by A.~Douglis and L.~Nirenberg, \cite{DN} (see also, for 
instance, \cite[Ch. 1, \S 3]{EgShu} or, \cite[\S 9.2]{WRL}). Namely, let the entries 
of an $(k\times k)$-matrix linear operator $A$ be scalar differential operators 
$A^{(p,r)} = \sum_{|\alpha| \leq m} a_{\alpha}^{(p,r)} (x) \partial ^\alpha$ with 
$a_{\alpha}^{(p,r)} (x)$ being the components of the functional $(k\times k)$-matrix $a_{\alpha}^{(p,r)} (x)$. 
Given two vectors $\vec s, \vec t\in {\mathbb R}^k$, the $(\vec s, \vec t)$-principal part 
of the operator $A$ is the $(k\times k)$-matrix linear operator $\tilde A$ with components 
$$
\tilde A^{(p,r)}= \left\{ 
\begin{array}{ll}
\sum_{|\alpha| =s_p-t_r} a_{\alpha}^{(p,r)} (x) \partial ^\alpha , & s_p\geq t_r, \\ 0,  & s_p< t_r. \\
\end{array}
\right.
$$
Then $(\vec s, \vec t)$-principal symbol of $A$ is the $(k\times k)$-matrix 
$\sigma_{\vec s, \vec t} (X) (x,\zeta)$ with the components 
$$
\Big(\sum_{|\alpha| =s_p-t_r} a_{\alpha}^{(p,r)} (x) \zeta ^\alpha \Big) .
$$
The operator $A$ is called Douglis-Nirenberg elliptic, if 
there are two vectors $\vec{s}, \vec{t}\in {\mathbb Z}^k$ such that 
$$
\det \sigma_{\vec{s}, \vec{t}} (X) (x,\zeta) \ne 0
\mbox{ for all } x\in X, \zeta \in {\mathbb R}^n \setminus \{0\}.
$$
Next, for the \textit{principal symbols} of the operators from  complex 
\eqref{eq.complex}, we have 
\begin{equation}
\label{eq.complex.symb.rel}
\sigma(A_{q+1}) \circ 
\sigma(A_q) \equiv 0. 
\end{equation}
Complex \eqref{eq.complex}  is called elliptic, 
 if the corresponding symbolic complex, 
\begin{equation*} 
0\rightarrow \pi^* E_0\stackrel{\sigma(A_0)}{\rightarrow} \pi^* E_1\stackrel{ \sigma 
(A_1)} {\rightarrow} \pi^* E_2 \rightarrow \dots  
\stackrel{ \sigma (A_{N-1})} {\rightarrow}  \pi^* E_{N} \rightarrow 0,
\end{equation*}
is exact for all $(x,z) \in T^{\ast} X \setminus \{0\} $, i.e. the  range of the mapping 
$\sigma (A_q)$ coincides with the kernel of the mapping $\sigma (A_{q+1})$. In particular, 
$\sigma (A_0)$ is 
injective and $\sigma (A_{N-1})$ is surjective for  all $(x,z) \in T^{\ast} X \setminus \{0\} $. 
Of course, an operator $A_0$ is elliptic if and only if the following complex  is elliptic: 
\begin{equation} \label{eq.complex.symb.short}
0\rightarrow \pi^* E_0\stackrel{\sigma(A_0)}{\rightarrow} \pi^* E_1 \rightarrow 0.
\end{equation} 
There is also a  Douglis-Nirenberg type ellipticity for elliptic complexes, see \cite{AM}.

For the sake of notations, we set $\sigma_q = \sigma (A_q)$ and $\delta_q =  \sigma^*_q \sigma _q + 
\sigma_{q-1} \sigma ^*_{q-1}$; then $\sigma^*_q = \sigma (A^*_q)$ and 
according to \eqref{eq.complex.symb.rel}, for all $0\leq q \leq N-1$  we have 
\begin{equation} \label{eq.Laplace.symb.rel}
\sigma^*_{q} \, \sigma^*_{q+1} = 0, \, 
\delta_{q+1} \, \sigma_q = \sigma_q \, \delta_{q} = \sigma_j \, \sigma_q^*\, \sigma_q, 
\, \sigma^*_q \,\delta_{q+1}  =  \delta_{q} \,\sigma^*_j =    \sigma_q^* \, \sigma_q \, \sigma^*_q.
\end{equation}

\begin{lem} \label{l.exact}
Complex \eqref{eq.complex} is elliptic  if and only if the mappings  
$\delta_q: \pi^* E_q \to \pi^* E_q $ 
are bijective for all $(x,z) \in T^{\ast} X \setminus \{0\} $ and all $0\leq q \leq N$.
\end{lem}

Denote by $\Delta_q$ the Hodge's Laplacians of complex \eqref{eq.complex}:
$$
\Delta_q = A_q^{*} A_q + A_{q-1} A_{q-1}^{*} ,\, 0\leq q \leq N.
$$ 
If $m_q=m_{q-1}$ then $\delta_q =\sigma(A_q^{*} A_q + A_{q-1} A_{q-1}^{*}) $. 
According to Lemma \ref{l.exact}, if complex \eqref{eq.complex} satisfies \eqref{eq.all.equal} then, 
for the complex to be elliptic, it is necessary and sufficient that 
the Laplacians $\Delta_q$ of the complex are strongly elliptic differential 
operators of order $2m$ for all $0\leq q \leq N$. 

Next, given a pair $\boldsymbol{\mu}_q$ consisting of formally non-negative self-adjoint  
differential operators $\boldsymbol{\mu}^{(0)}_q \in \mathrm{Diff}_{2\tilde m_q} 
(X, E_{q+1}\to E_{q+1})$ and $\boldsymbol{\mu}^{(1)}_q 
\in \mathrm{Diff}_{2\hat m_q}(X, E_{q-1}\to E_{q-1}) $, $0\leq q \leq N$, 
 with some  numbers $\tilde m_q, \hat m_q\in {\mathbb Z}_+$, satisfying 
$0 \leq  \tilde m_q \leq m-m_q$, $0\leq \hat m_q \leq m-m_{q-1}$, 
we denote by $\Delta_{q,\boldsymbol{\mu}}$ the steady Lam\'e type operators 
$$
\Delta_{q,\boldsymbol{\mu}} = A^*_{q} \boldsymbol{\mu}_q ^{(0)}A_q + A_{q-1} 
\boldsymbol{\mu}_{q} ^{(1)}A^*_{q-1} .
$$
If orders $\tilde m_q$ and $\hat m_q$ equal to zero then  strong ellipticity means 
that $\boldsymbol{\mu}_{q} ^{(0)} $, $\boldsymbol{\mu}_{q+2} ^{(1)}$ are bijective 
self-adjoint non-negative mappings. 
In general, we may produce the operators $\boldsymbol{\mu}_q ^{(0)}$, $\boldsymbol{\mu}_q ^{(1)}$ with the use of connections over $E_{q+1}$ and $E_{q-1}$, respectively: 
$$
\boldsymbol{\mu}_q ^{(0)} = (\nabla_{E_{q+1}}^* \nabla_{E_{q+1}})^{\tilde m_q}, \, 
\boldsymbol{\mu}_q ^{(0)} = (\nabla_{E_{q-1}}^* \nabla_{E_{q-1}})^{\hat m_q}.
$$
 On this way, taking $\tilde m=m-m_q$, $\hat m=m-m_{q-1}$  we may achieve 
that all the operators $\Delta_{q,\boldsymbol{\mu}}$ have the same order $2m$. 
For this reason we will often use the following assumption.

\begin{asm}
\label{asm.mu}
The formally self-adjoint non-negative operators  $\boldsymbol{\mu}^{(0)}_{j}$ and $\boldsymbol{\mu}^{(1)}_{j}$ 
are strongly elliptic. 
\end{asm}

Set $\delta_{q,\boldsymbol{\mu}} = \sigma^*_q \, \sigma(\boldsymbol{\mu}^{(0)}_{q})\, \sigma _q + \sigma_{q-1} 
\,\sigma(\boldsymbol{\mu}^{(1)}_{q}) \, \sigma ^*_{q-1}$. 
\begin{lem} \label{l.exact.mu}
Let  complex \eqref{eq.complex} be elliptic. If $0\leq j \leq N$ then, under Assumption {\rm \ref{asm.mu}}, the mapping  
$\delta_{j,\boldsymbol{\mu}}: \pi^* E_j \to \pi^* E_j $ 
is bijective for all $(x,z) \in T^{\ast} X \setminus \{0\}$; in particular, the operator 
$\Delta_{j,\boldsymbol{\mu}}$ is strongly elliptic self-adjoint non-negative, too, if 
we additionally have $m_j + \tilde m_j = m_{j-1} + \hat m_j$.
\end{lem}

\begin{rem} Clearly, the generalized Laplacians can be factorized as follows
\begin{equation} \label{eq.factor}
\Delta_{q,\boldsymbol{\mu}} =  
\left(
\begin{array}{ll}  A^*_{q} , &  
  A_{q-1}  \boldsymbol{\mu}_{q} ^{(1)} \\
\end{array}
\right) \left(
\begin{array}{ll} \boldsymbol{\mu}_q ^{(0)} A_{q} \\ 
  A^*_{q-1}  \\
\end{array}
\right),  0\leq q \leq N.
\end{equation}
As we have noted above, a compatibility complex $\{ A_q , E_q\}$ for an 
operator $A_0$ is not uniquely defined. For this reason, formula \eqref{eq.factor} suggests 
the following natural conditions for the operators $\boldsymbol{\mu}_q ^{(k)}$:
\begin{equation} \label{eq.coh}
A_{q+1}\, \boldsymbol{\mu}_{q+2} ^{(1)}\ \boldsymbol{\mu}_{q} ^{(0)} \,
A_q \equiv  0  \mbox{ for all }  0\leq q\leq N-1.
\end{equation} 
On the symbolic level this means
\begin{equation} \label{eq.coh.symb}
\sigma_{q+1} \, \sigma(\boldsymbol{\mu}_{q+2} ^{(1)}) \, \sigma ( \boldsymbol{\mu}_{q} ^{(0)}) \,
\sigma_q  \equiv  0  \mbox{ for all }  0\leq q\leq N-1.
\end{equation}
If the complex $\{ A_q, E_q\} $ consists of the operators  with constant coefficients 
then we may set $\boldsymbol{\mu}_{q} ^{(0)} 
= (-1)^{\tilde m_q} I_{k_q+1}\otimes \Delta^{\tilde m_q}$, $\boldsymbol{\mu}_{q} ^{(1)} 
= (-1)^{\hat m_q} I_{k_q-1}\otimes \Delta^{\hat m_q}$; in this case, 
$$
\boldsymbol{\mu}_{q} ^{(0)} A_q = (-1)^{\tilde m_q} A_q (I_{k_q}\otimes \Delta^{\tilde m_q}), \,  
 A_{q+2} \boldsymbol{\mu}_{q+2} ^{(1)} = (-1)^{\hat m_q}  (I_{k_{q+3}}\otimes \Delta^{\hat m_q}) A_{q+2},
$$
and hence \eqref{eq.coh}, \eqref{eq.coh.symb} hold true. In general case we should look for suitable 
commutative relations between $\boldsymbol{\mu}_{q} ^{(0)} $, $\boldsymbol{\mu}_{q+2} ^{(1)}$, $A_q$ 
and $A_{q+2}$ that is not a trivial task. 
\end{rem}

Finally, if $m_q+\tilde m_q =m_{q-1} +\hat m_q$, we may consider the Helmholtz-Lam\'e operators 
\begin{equation} \label{eq.diagonal} 
{\mathfrak D}_{q,\boldsymbol{\mu}}= \Delta_{q,\boldsymbol{\mu}} + \sum_{|\alpha|\leq 2(m_q+\tilde m_q)-1}
d_\alpha (x) \partial^\alpha
\end{equation}
that are strongly elliptic if the operators $\boldsymbol{\mu}_{q}^{(0)}$, $\boldsymbol{\mu}_{q}^{(1)}$ are 
strongly elliptic (${\mathfrak D}_{q}$ corresponds to the case where $\boldsymbol{\mu}_{q} ^{(0)} 
=  I_{k_q+1}$, $\boldsymbol{\mu}_{q} ^{(1)} = I_{k_q-1}$). 

Note that the most natural part of the low order perturbation of the Laplacian  
$\Delta_{q,\boldsymbol{\mu}}$ is   usually  given as follows:
\begin{equation} \label{eq.diagonal.low}
C_q A_q + \tilde C_q A^*_{q-1} + M_q 
\end{equation}
with a formally self-adjoint  operator $M_q \in \mathrm{Diff}_{0}(X,E_{q} \to E_{q})$,  and operators $C \in
\mathrm{Diff}_{m_q+2\tilde m_q-1}(X,E_{q+1} \to E_{q})$, 
$\tilde C \in \mathrm{Diff}_{m_{q-1}+2\hat m_q-1}(X,E_{q-1} \to E_{q})$. 

\subsection{The induced complexes and time dependent processes}

The constructions considered in the previous subsection  
are fit for steady models of Mathematical Physics. 
To use the differential complexes for time dependent models, one 
may introduce the so-called induced complex
\begin{equation*} 
0\rightarrow C^\infty(X, E_0 (t)) \stackrel{A_{0}}{\rightarrow}
\dots \rightarrow \dots 
\stackrel{A_{N-1}}{\rightarrow} C^\infty(X, E_N (t)) \rightarrow 0, 
\end{equation*}
where sections of the induced bundles $E_q (t)$ and the coefficients $a^{(q)}_\alpha$ of the 
differential operators $A_q$ depend on both $x$ and the real parameter $t$. 
The induced complex 
$\{A_q,E_q (t)\}$  is elliptic on $X \times \mathbb [0,T)$ with a (possibly, infinite) 
time $T$, if the corresponding symbolic complex, 
\begin{equation*} 
0\rightarrow 
\pi^* E_0 (t)\stackrel{\sigma(A_0)}{\rightarrow} \pi^* E_1 (t)\stackrel{ \sigma 
(A_1)} {\rightarrow} \pi^* E_2 (t)\rightarrow \dots  
\stackrel{ \sigma (A_{N-1})} {\rightarrow}  \pi^* E_{N} (t)\rightarrow 0,
\end{equation*}
is exact for all $(x,z) \in T^{\ast} X \setminus \{0\} $ and each $t \in [0,T)$. Note that for some kind of problems one needs a more subtle 
notion of  ellipticity with a parameter, see, for instance, \cite{AV} or 
\cite[Ch 2, \S 2]{EgShu}.

In any case, one may easily introduce the  operators 
$$
{\mathcal L}_{q,\boldsymbol{\mu}} = \partial _t + {\mathfrak D}_{q,\boldsymbol{\mu}}, \, \, 
{\mathcal H}_{q,\boldsymbol{\mu}} = \partial^2 _t + {\mathfrak D}_{q,\boldsymbol{\mu}}
$$ 
over $X \times [0,T)$. Note  that ${\mathcal L}_{q,\boldsymbol{\mu}}$  is strongly parabolic, if 
${\mathfrak D}_{q,\boldsymbol{\mu}}$ is strongly elliptic,  see \cite{eid}, 
\cite[Ch 1, \S 3, Ch. 2, \S 5]{EgShu} and  ${\mathcal H}_{q,\boldsymbol{\mu}}$ has the 'hyperbolicity' 
properties if ${\mathfrak D}_{q,\boldsymbol{\mu}}$ is strongly elliptic, see, for instance, 
\cite[Ch 1, \S 3, Ch. 2, \S 4]{EgShu}. 

\section{Maxwell's and Stokes' type operators for differential complexes}
\label{s.MS}

It is well known that  Stokes' system  $S = S_{1,\mu} (d, \partial_t)$, 
\begin{equation} \label{eq.Stokes.class}
S  \, \left( \begin{array}{lll} 
\vec{v} \\  p \\
\end{array}
\right) 
= \left( \begin{array}{lll} 
(\partial_t - \mu \Delta)I_n & \nabla \\ {\rm div} & 0 \\
\end{array}
\right) \left( \begin{array}{lll} 
\vec{v} \\  p \\
\end{array}
\right) = 
\left( \begin{array}{lll} 
\vec{f} \\  0 \\
\end{array}
\right) 
\end{equation}
plays an essential role in mathematical models for incompressible fluid with given 
the dynamical viscosity $\mu>0$ of the fluid under the consideration, the density vector 
of outer forces $\vec{f}$, the search-for velocity vector field $\vec{v}$ 
and  pressure $p$ of the flow, see, for instance, 
\cite{St}, \cite{LaLi}, \cite{Tema79}.
Actually, Stokes'  system $S$ and its steady version give the principal linear parts of the 
steady and evolutionary 
Navier-Stokes equations in ${\mathbb R}^n$, $n\geq 2$, 
endowed with the non-linear perturbation  given by 
\begin{equation*} 
N_1(\vec{v}) = \left( \begin{array}{ccc} 
(\vec{v}  \cdot \nabla) \vec{v} & 0\\ 0 & 0 \\
\end{array}
\right) = \left( \begin{array}{ccc} 
\Big(\sum\limits_{j=1}^n v_j \partial_j \Big) \vec{v} & 0\\ 0 & 0 \\
\end{array}
\right) .
\end{equation*}
It appears, that $S$ can be easily written in the context of the de Rham complex 
$\{ \nabla, \mathrm{curl}, 
\mathrm{div}\} $ in ${\mathbb R}^3$ at the step $q=1$ with 
$$
\Delta_1 = 
-\Delta I_3 = \mathrm{curl} ^* \mathrm{curl} + \mathrm{div}^*  \mathrm{div} = 
 \mathrm{curl}  \, \mathrm{curl} - \nabla  \, \mathrm{div}. 
$$  
A  generalisation of Stokes' type operators for elliptic complex \eqref{eq.complex} was proposed in 
\cite[formula (1.2)]{MeShT} (cf. \cite{ShTaSEMR} for the de Rham complex):
$$
\tilde S_q (A ,\partial _t) = \left( \begin{array}{ccc} 
{\mathcal L}_{q,\boldsymbol{\mu}}  & A_{q-1} \\ A_{q-1}^* & 0 \\
\end{array}
\right), \, 0\leq q \leq N.
$$
It was noted in \cite[Proposition 2.2]{ShTaSEMR}  and \cite[formula (0.4)]{PaSh} that for $q\geq 2$ 
one more natural line should be added to Stokes' type operator associated with complex 
\eqref{eq.complex} that is missed for $q=1$:
$$
\hat S_q (A,\partial_t ) = \left( \begin{array}{ccc} 
{\mathcal L}_{q,\boldsymbol{\mu}}  & A_{q-1} \\ A_{q-1}^* & 0 \\
 0 & A_{q-2}^*  \\
\end{array}
\right) , \, 0\leq q \leq N.
$$
Actually the additional operator equation provides some uniqueness for Stokes' type equations. However, adding 
this natural equation we see that new operators becomes overdetermined. This fact may complicate essentially  the 
theory of Stokes' type equations for the complex \eqref{eq.complex}  at the degrees $q>1$. To overcome this 
difficulty, let us introduce  slightly different generalisations of  Stokes' type operators. 

\subsection{Steady Maxwell's and Stokes' type operators for elliptic complexes}
Consider the following steady Maxwell's and Stokes' type operators related to complex \eqref{eq.complex} at 
degrees $0\leq q\leq N$. Namely, set $r_q=(\sum_{j=0}^q k_j)$, $0\leq q \leq N$. 
By $B^N $ we denote 
$(N+1)\times 
(N+1)$-block matrix (in fact, it is a $(r_N \times r_N)$-matrix), 
such that each its block $b^N _{ij}$ is a 
$(k_{N-i+1} \times k_{N-j+1})$-matrix. 
Let $B_j$ be such a matrix with $b^N _{jj} = I_{k_j}$ and $B^N _{pq} = 0$
for $p\ne j$ or $q \ne j$. 
Clearly,
\begin{equation} \label{eq.Bi}
B_j B_j = B_j, \,\, B_i B_j = 0 \mbox{ if } j\ne i 
\end{equation} 
and all the blocks of the matrix $ B_i B^N B_j $ equal 
to zero except the block 
$$
(B_i B^N B_j)_{ij}=b^N_{ij}.
$$ 
For this reason, if $P$ is an operator  of type $E_{N-i+1} \to E_{N-j+1}$  then we denote by 
$B_i P B_j$ the $(N+1)\times (N+1)$-block matrix with all the block being zero except the block $b^N_{ij}=P$.

Next, set ${\mathfrak E}_q = \oplus_{j=0}^q E_j$, $0\leq q \leq N$. 
Given set $\boldsymbol{\mu}= (\boldsymbol{\mu}_0, \dots, \boldsymbol{\mu}_q)$ of pairs of 
differential operators,  we introduce the following Maxwell's type operators 
acting on sections $U_N =  (u_0, \dots, u_N)$ 
of the bundle ${\mathfrak E}_N$: 
${\mathcal M}^{(0)}_{0,\boldsymbol{\mu}} (A) =0$, 
${\mathcal M}^{(1)}_{0,\boldsymbol{\mu}} (A) =0$, 
\begin{equation} \label{eq.M.0} 
{\mathcal M}^{(0)}_{q,\boldsymbol{\mu}} (A)=
\sum_{j=0}^{q-1} \Big( B_{j+1} \boldsymbol{\mu}_j ^{(0)} A_{j}   B_{j} + 
B_{j}  A^*_{j}  B_{j+1} \Big) , \, 1\leq q \leq N,
\end{equation}
\begin{equation} \label{eq.M.1}  
{\mathcal M}^{(1)}_{q,\boldsymbol{\mu}} (A)=
\sum_{j=0}^{q-1} \Big( B_{j+1}  A_{j} \boldsymbol{\mu}_{j+1} ^{(1)} B_{j} + 
B_{j}  A^*_{j}  B_{j+1} \Big) , \, 1\leq q \leq N.
\end{equation}
Obviously, ${\mathcal M}^{(0)}_{q,\boldsymbol{\mu}} (A) = 
{\mathcal M}^{(1)}_{q,\boldsymbol{\mu}} (A)$, if  
\begin{equation} \label{eq.commute.mu}
A_j \boldsymbol{\mu}_{j+1}^{(1)} =  \boldsymbol{\mu}_{j}^{(0)} A_j  
\mbox{ for all } 0\leq j\leq q;
\end{equation}
consequently, we will write ${\mathcal M}_{q,\boldsymbol{\mu}} (A)$ for 
${\mathcal M}^{(j)}_{q,\boldsymbol{\mu}} (A)$ in this case. 
In particular, we will use the notation  ${\mathcal M}_{q} (A) $ 
in the simplest case where $\boldsymbol{\mu}_j ^{(0)} = 
I_{k_{j+1}}$, $\boldsymbol{\mu}_j ^{(1)} = 
I_{k_{j-1}}$ 
for all $0\leq j\leq q$. 

Similarly, we introduce  Stokes' type operators 
\begin{equation} \label{eq.Stokes.full.short}
S_{q,a} (A,{\mathfrak D}_{\boldsymbol{\mu}}) 
= \sum_{j=0}^q B_j {\mathfrak D}_{j,\boldsymbol{\mu}} 
B_j + a \, \sum_{j=0}^{q-1} \Big( B_{j+1} A_{j} B_{j} + B_{j} A^*_{j} B_{j+1} \Big),
\end{equation}
with $a=a_q$ being equal to $0$ or $1$. 

In fact, by the construction, the operators 
${\mathcal M}^{(i)}_{q,\boldsymbol{\mu}} (A)$, $S_{q,1} (A,{\mathfrak D}_{\boldsymbol{\mu}})$ 
act on sections $U_q =  (u_0, \dots, u_q)$ of the bundle ${\mathfrak E}_q$ and hence we will identify them  
with the lower right $(r_q\times r_q)$-minors of the related full $(r_N\times r_N)$-matrices. 
For instance, in a more bulky 
matrix form the operator $S_{q,1} (A,{\mathfrak D}_{\boldsymbol{\mu}}) $ may be written as  
\begin{equation*} 
\left( \begin{array}{cccccccc}  
{\mathfrak D}_{q,\boldsymbol{\mu}} 
& A_{q-1} & 0 & 0 & 0 & 0 & \dots & 0 \\ 
A_{q-1}^* & {\mathfrak D}_{q-1,\boldsymbol{\mu}} & A_{q-2} & 0 & 0 & 0 & \dots & 0\\
 0 & A_{q-2}^*  & {\mathfrak D}_{q-2,\boldsymbol{\mu}}
 &  A_{q-3} & 0 & 0 & \dots & 0  \\
 \dots & \dots  & \dots &  \dots & \dots & \dots & \dots  & \dots\\
 0 & \dots  & \dots & 0 &  A_2^* & 
{\mathfrak D}_{2,\boldsymbol{\mu}} & A_1 & 0  \\
 0 & \dots  & \dots & \dots & 0 & A_1^* &  
{\mathfrak D}_{1,\boldsymbol{\mu}}   & A_0  \\
 0 & \dots & \dots  & \dots & \dots & 0 & A_0^* &    
{\mathfrak D}_{0,\boldsymbol{\mu}}  \\
\end{array}
\right) .
\end{equation*}
More compact notations \eqref{eq.M.0}, \eqref{eq.M.1}, \eqref{eq.Stokes.full.short} echo with the sedeonic form  
of equations of Mathematical Physics 
proposed in 
\cite{Mir15}, \cite{Mir18}, \cite{Mir20A}, \cite{Mir20B}. 

Let's explain the connection between  Maxwell's and Stokes' type operators.

\begin{lem} 
\label{l.factor}
If \eqref{eq.coh} is fulfilled for all $0\leq j \leq q$, then we have
\begin{equation} \label{eq.P.mu}
{\mathcal M}^{(1)}_{q,\boldsymbol{\mu}} (A){\mathcal M}^{(0)}_{q,\boldsymbol{\mu}} (A) = B_{q} A_{q-1} 
\boldsymbol{\mu}_{q-1} ^{(1)} A^*_{q-1} B_{q} + 
\sum_{j=0}^{q-1} B_j \Delta_{j,\boldsymbol{\mu}} B_j .
\end{equation}
In particular, if ${\mathfrak D}_{j,\boldsymbol{\mu}_{j}} = \Delta_{j,\boldsymbol{\mu}_{j}}$ for 
all  $0\leq j \leq N$, then 
\begin{equation*} 
S_{q,a} (A,{\mathfrak D}_{\boldsymbol{\mu}}) 
=  {\mathcal M}^{(1)}_{q,\boldsymbol{\mu}} (A){\mathcal M}^{(0)}_{q,\boldsymbol{\mu}} (A) + 
B_{q} A^*_{q} \boldsymbol{\mu}_{q} ^{(0)} A_{q} B_{q} + a \, {\mathcal M}_{q} (A) ,
\, 0\leq q \leq N.
\end{equation*}
\end{lem}

Next, we note that Stokes' type operator $S_{0} (A,{\mathfrak D}_{0,\boldsymbol{\mu}}) = 
{\mathfrak D}_{0,\boldsymbol{\mu}}$ is elliptic and strongly elliptic on $X$ 
if the operator $\boldsymbol{\mu}_0^{(0)}$ is strongly elliptic. 
As it is known, see, for instance, \cite[Ch. II, \S 4, Example 1]{Mas97}, 
Stokes' operator \eqref{eq.Stokes.class} 
is Douglis-Nirenberg elliptic over ${\mathbb R}^n$. Then the following two statements are rather expectable. 
To formulate the statements, we set 
$$
\tilde \sigma ({\mathcal M}^{(0)}_{q,\boldsymbol{\mu}} (A)) =
\sum_{j=0}^{q-1} \Big( B_{j+1}\, \sigma (
\boldsymbol{\mu}_j ^{(0)}) \, \sigma_{j} \, B_{j} + 
B_{j} \, \sigma^*_{j} \, B_{j+1} \Big),
$$
$$
\tilde \sigma ({\mathcal M}^{(1)}_{q,\boldsymbol{\mu}} (A)) =
\sum_{j=0}^{q-1} \Big( B_{j+1}  \, \sigma_{j} \, \sigma (\boldsymbol{\mu}_{j+1} ^{(1)})\,  B_{j} + 
B_{j} \, \sigma^*_{j} \, B_{j+1} \Big).
$$

\begin{prop} \label{p.symb.inverse.0}
Let complex \eqref{eq.complex} be elliptic, Assumption {\rm \ref{asm.mu}} be fulfilled for all $0\leq j \leq N$ 
and \eqref{eq.coh.symb} be true for all $0\leq j \leq N$. 
Then the symbolic matrices $\tilde \sigma ({\mathcal M}^{(i)}_{N,\boldsymbol{\mu}} (A))$, $i=0,1$, are invertible 
for all $(x,\zeta) \in T^{\ast} X \setminus \{ 0\}$. 
In particular, 
\begin{itemize}
\item
the operators ${\mathcal M}^{(0)}_{N,\boldsymbol{\mu}} (A)$, ${\mathcal M}^{(1)}_{N,\boldsymbol{\mu}} (A)$, 
$N\geq 1$, are  elliptic if 
\eqref{eq.all.equal} is true for all $0\leq j \leq N$; 
\item  
the operator $S_{N,a}(A,\mathfrak{D}_{\boldsymbol{\mu}})$, $N\geq 1$,  is elliptic for any $a$ if 
for all  $0\leq j \leq N$
\begin{equation} \label{eq.orders.mu}
m_j + \tilde m_j =m_{j-1} + \hat m_j =m ; 
\end{equation}
\item 
under \eqref{eq.all.equal}  
the operator $S_{N,1}(A,\mathfrak{D}_{\boldsymbol{\mu}})$, $N\geq 1$,   is elliptic  if 
\begin{equation} \label{eq.all.zero}
\boldsymbol{\mu}^{(0)}_j =0 \mbox{ for all } 0\leq j \leq N-1, \,\, 
\boldsymbol{\mu}^{(1)}_j =0  \mbox{ for all } 1\leq j \leq N. 
\end{equation}
\item
the operators ${\mathcal M}^{(0)}_{N,\boldsymbol{\mu}} (A)$, 
${\mathcal M}^{(1)}_{N,\boldsymbol{\mu}} (A)$,  $N\geq 1$,  
are Douglis-Nirenberg elliptic; 
\item
the operator $S_{N,1} (A,\mathfrak{D}_{\boldsymbol{\mu}})$, $N\geq 1$,  
is  Douglis-Nirenberg elliptic.  
\end{itemize}
 \end{prop}

\begin{proof} Indeed,  similarly to \eqref{eq.P.mu}, 
under condition \eqref{eq.coh.symb}, for 
all $0\leq q \leq N$ we have 
\begin{equation} \label{eq.symb.P}
\tilde \sigma({\mathcal M}_{q,\boldsymbol{\mu}}^{(1)}  (A)) \tilde  
\sigma ({\mathcal M}^{(0)} _{q,\boldsymbol{\mu}} (A))
 = B_{q} \sigma_{q-1} \boldsymbol{\mu}_{q} ^{(1)} \sigma^*_{q-1} B_{q} + 
\sum_{j=0}^{q-1} B_j\delta_{j,\boldsymbol{\mu}_j} B_j .
\end{equation}
In particular, 
\begin{equation} \label{eq.symb.P.N}
\tilde \sigma ({\mathcal M}^{(1)} _{N,\boldsymbol{\mu}} (A)) 
\tilde \sigma ({\mathcal M}^{(0)} _{N,\boldsymbol{\mu}} (A)) 
=\sum_{j=0}^{N} B_j \delta_{j,\boldsymbol{\mu}_j} B_j.  
\end{equation}
Hence, as the symbolic matrices $\delta_{j,\boldsymbol{\mu}_j}$ are invertible 
for all $(x,\zeta) \in T^{\ast} X $ with $\zeta \ne 0$ and all $0\leq j \leq N$ 
(see Lemma \ref{l.exact.mu}), then the matrices $\tilde \sigma ({\mathcal M}^{(0)}_{N,\boldsymbol{\mu}} (A))$, 
$\tilde \sigma ({\mathcal M}^{(1)}_{N,\boldsymbol{\mu}} (A))$ are invertible for 
such $(x,\zeta)$, too.

If all the operators $A_q$, $0\leq q \leq N-1$, have the same order $m$ then 
the orders of the operators $\boldsymbol{\mu}^{(i)}_j$ equal to zero and hence 
$$
\tilde \sigma ({\mathcal M}^{(0)}_{N,\boldsymbol{\mu}} (A)) = \sigma ({\mathcal M}^{(0)}_{N,\boldsymbol{\mu}} 
(A)), \, \, 
\tilde \sigma ({\mathcal M}^{(1)}_{N,\boldsymbol{\mu}} (A)) = \sigma ({\mathcal M}^{(1)}_{N,\boldsymbol{\mu}} 
(A)),
$$ 
i.e. the operators ${\mathcal M}^{(0)}_{N,\boldsymbol{\mu}} (A)$
${\mathcal M}^{(1)}_{N,\boldsymbol{\mu}} (A)$ are elliptic.

Moreover, under \eqref{eq.orders.mu}, 
$$
\sigma (S_{N,a} (A,\mathfrak{D}_{\boldsymbol{\mu}})) = \sum_{j=0}^{N} B_j \delta_{j,\boldsymbol{\mu}} B_j
$$
and then $S_{N,a} (A,\mathfrak{D}_{\boldsymbol{\mu}})$  is elliptic for any $a$ because of Lemma 
\ref{l.exact.mu}.

If \eqref{eq.all.equal} is fulfilled then 
$$
\sigma (S_{N,a} (A,\mathfrak{D}_{\boldsymbol{\mu}})) = a \, \sigma ({\mathcal M}_{N, \boldsymbol{\mu}} (A)).
$$
Thus, in this case the operator 
$S_{N,a} (A,\mathfrak{D}_{\boldsymbol{\mu}})$ is elliptic, too, if $a =1$. 

If $N=1$ (that corresponds to an 
elliptic operator with symbolic complex \eqref{eq.complex.symb.short}) then  
operators ${\mathcal M}^{(0)}_{1,\boldsymbol{\mu}} (A)$, ${\mathcal M}^{(1)}_{1,\boldsymbol{\mu}} 
(A)$ are always elliptic
because the orders of the operators $\boldsymbol{\mu}^{(i)}_0$ equal to zero. 

If $N\geq 2$ and the orders $m_q$ of the operators $A_q$ are different, then, 
we may solve the following system of $4N$ equations with respect to $4(N+1)$ unknown numbers  
$s^{(i)}_1, \dots s^{(i)}_{N+1}, t^{(i)}_1, \dots t^{(i)}_{N+1}$, $i=0,1$:
\begin{equation} \label{eq.DN.P}
\left\{ 
\begin{array}{llll} 
s^{(0)}_j- t^{(0)}_{j+1} = m_{N-j}+\tilde m_{N-j}, &  s^{(0)}_{j+1}- t^{(0)}_{j} = m_{N-j}, 
 &  1\leq j \leq N, \\
s^{(1)}_{j+1}- t^{(1)}_{j} = m_{N-j} +\hat m_{N-j},  &  s^{(1)}_j- t^{(1)}_{j+1} = m_{N-j},  &  1\leq j \leq N,  \\
\end{array}
\right.
\end{equation} 
As $4(N+1)-4N=4$ we set $t^{(0)}_1=t^{(0)}_2=t^{(1)}_1=t^{(1)}_2=0$ and then 
$$
s^{(0)}_1=m_{N-1}+\tilde m_{N-1}, \, s^{(0)}_2=m_{N-1}, 
\, s^{(1)}_1=m_{N-1}, \, s^{(1)}_2=m_{N-1} +\hat m_{N-1},
$$
and we obtain a recurrent formula:
\begin{equation*} 
\left\{ 
\begin{array}{llll} 
t^{(0)}_{j+1} =  s^{(0)}_j- m_{N-j} -\tilde m_{N-j}, & s^{(0)}_{j+1} = m_{N-j}+ t^{(0)}_{j},  & 2\leq j \leq N,  \\
t^{(1)}_{j+1} =  s^{(1)}_j- m_{N-j} , & s^{(1)}_{j+1} = m_{N-j}+ t^{(1)}_{j}-\hat m_{N-j},  & 2\leq j \leq N.  \\
\end{array}
\right.
\end{equation*} 
Hence we obtain a solution $\vec{s}^{(0)}$, $\vec{t}^{(0)}, \vec{s}^{(1)}, 
\vec{t}^{(1)}$ to system  \eqref{eq.DN.P} with  integer components. 
Then there is a non-negative integer $c$ 
such that the vectors $\vec{s}^{(j)}=(s_1^{(j)}+c, \dots s^{(j)}_{N+1}+c)$, $\vec{t}^{(j)}=
(t_1^{(j)}+c, \dots t_{N+1}^{(j)}+c)$, $j=0,1$, 
are solutions to \eqref{eq.DN.P} with non negative components. 
Assigning the values $s^{(i)}_p$, $t^{(i)}_r$ for each component of the block 
${\mathcal M}^{(i)}_{N,\boldsymbol{\mu}} (A,p,r)$ in the block matrix 
${\mathcal M}^{(i)}_{N,\boldsymbol{\mu}} (A)$ we see that
$$
\tilde \sigma ({\mathcal M}^{(0)}_{N,\boldsymbol{\mu}} 
A))= \sigma _{\vec{s^{(0)}}, \vec{t}^{(0)}}({\mathcal M}^{(0)}_{N,\boldsymbol{\mu}} (A)), \,\,
\tilde \sigma ({\mathcal M}^{(1)}_{N,\boldsymbol{\mu}} 
A))= \sigma _{\vec{s}^{(1)}, \vec{t}^{(1)}}({\mathcal M}^{(1)}_{N,\boldsymbol{\mu}} (A)).  
$$
Thus, Lemma  \ref{l.exact.mu} and formula \eqref{eq.symb.P.N} imply that the operators 
${\mathcal M}^{(0)}_{N,\boldsymbol{\mu}}$, 
${\mathcal M}^{(1)}_{N,\boldsymbol{\mu}}$ are Douglis-Nirenberg elliptic. In particular, 
${\mathcal M}_{N,\boldsymbol{\mu}} $ is Douglis-Nirenberg elliptic, too. 

And, finally, under \eqref{eq.all.zero} we have
$$
\sigma _{\vec{s}^{(0)}, \vec{t}^{(0)}} (S_{N,1} (A,\mathfrak{D}_{\boldsymbol{\mu}}))  = 
\tilde \sigma ({\mathcal M}_{N} (A)) . 
$$
Therefore the operator  
$S_{N,1} (A, \mathfrak{D}_{\boldsymbol{\mu}})$ is always Douglis-Nirenberg elliptic.  
\end{proof}

\begin{prop} \label{t.stokes.ell} 
Let complex \eqref{eq.complex} be elliptic, \eqref{eq.coh.symb} be true, $N\geq 2$ and 
$0\leq q \leq N-1$. If \eqref{eq.orders.mu} 
and Assumption {\rm \ref{asm.mu}} are fulfilled for all 
$0\leq j \leq q$, then the Stokes operator $S_{q,a} (A,\mathfrak{D}_{\boldsymbol{\mu}})$ 
is (Petrovskii) elliptic.   If $m_q+\tilde m_q =m_{q-1}+ \hat m_q$,  
and  Assumption {\rm \ref{asm.mu}} is fulfilled for $j=q$
then $S_{q,1} (A,\mathfrak{D}_{\boldsymbol{\mu}})$ 
is a Douglis-Nirenberg elliptic operator.
\end{prop}

\begin{proof} 
For $q=0$ we always have $S_{0} (A, \mathfrak{D}_{\boldsymbol{\mu}})= \mathfrak{D}_{0,\boldsymbol{\mu}}$, 
i.e. it is strongly elliptic if the  
differential operator $\boldsymbol{\mu}^{(0)}_0$ is strongly elliptic  on $X$. 
Moreover, under the hypothesis of the first part of this proposition
we have
\begin{equation*} 
\sigma (S_{q,a} (A,{\mathfrak D}_{\boldsymbol{\mu}}) ) = \sum_{j=0}^q B_j \delta_{j, \boldsymbol{\mu}} 
B_j. 
\end{equation*}
As in  this particular case,  $ \Delta_{j, \boldsymbol{\mu}} $ are strongly elliptic 
operators, see Lemma \ref{l.exact.mu}, 
we conclude that the operator $S_{q,a} (A,{\mathfrak D}_{\boldsymbol{\mu}})$ is elliptic, too.

Let us prove the second statement of the proposition. With this purpose, let us solve 
the following system of $(2q+1)$ equations with respect to $2(q+1)$ unknown numbers  
$s_1, \dots s_{q+1}, t_1, \dots t_{q+1}$:
\begin{equation} \label{eq.DN.S}
\left\{ 
\begin{array}{lll} s_1-t_1 =2(m_q+\tilde m_q) =2(m_{q-1}+\hat m_q),\\ 
s_j- t_{j+1} = m_{q-j},\,\, s_{j+1}- t_{j} = m_{q-j},   & 1\leq j \leq q,  \\ 
\end{array}
\right.
\end{equation} 
As $2(q+1)-2q+1=1$, we set $t_1=0$ and then 
$$s_1=2(m_q+\tilde m_q), \,\, t_2=2 (m_q+\tilde m_q) -m_{q-1},
$$  
 and we again obtain a recurrent formula:
$$
t_{j+1} =  s_j- m_{N-j} , \,\, s_{j+1} = m_{N-j}+ t_{j},  \,\, 2\leq j \leq N. 
$$
Thus, system 
\eqref{eq.DN.S} has a solution $\vec{s}$, $\vec{t}$ with  integer components.  
Then there is a non-negative integer $c$ 
such that the numbers $s_1+c, \dots s_{q+1}+c, t_1+c, \dots t_{q+1}+c$
are solutions to \eqref{eq.DN.S} with non negative components.
Again, assigning values 
$s_p$, $t_r$ for each component of the block 
$S_{q,1} (A,{\mathfrak D}_{\boldsymbol{\mu}}),p,r)$ in the block matrix 
$S_{q,1} (A,{\mathfrak D}_{\boldsymbol{\mu}}))$ we see that 
$$
\sigma_{\vec{s},\vec{t}} (S_{q,1} (A,{\mathfrak D}_{\boldsymbol{\mu}})))= 
B_q \delta_{q,\boldsymbol{\mu}_q} B_q + \tilde \sigma({\mathcal M}_{q} (A)). 
$$
Next, using \eqref{eq.Laplace.symb.rel} we conclude that 
\begin{equation*} 
\delta_{j, \boldsymbol{\mu}} \sigma^*_j \sigma(\boldsymbol{\mu}^{(0)}_{j}) \sigma_j=
 \sigma^*_j \sigma(\boldsymbol{\mu}^{(0)}_{j}) \sigma_j  \sigma^*_j 
\sigma(\boldsymbol{\mu}^{(0)}_{j}) \sigma_j =  
 \sigma^*_j \sigma(\boldsymbol{\mu}^{(0)}_{j}) \sigma_j \delta_{j, \boldsymbol{\mu}}
\end{equation*}
and, if the matrix $\delta_{j, \boldsymbol{\mu}} $ is invertible, then 
\begin{equation} \label{eq.Laplace.mu.symb.rel.02}
 \sigma^*_j \sigma(\boldsymbol{\mu}^{(0)}_{j}) \sigma_j\delta^{-1}_{j, \boldsymbol{\mu}}= 
 \delta^{-1}_{j, \boldsymbol{\mu}} \sigma^*_j \sigma(\boldsymbol{\mu}^{(0)}_{j}) \sigma_j .
\end{equation}
Consider the following matrix:  
\begin{equation} \label{eq.Nq}
\mathcal{N}^{(q)}_\sigma = B_{q} 
\delta^{-1}_{q,\boldsymbol{\mu}} \sigma^*_{q} \sigma(\boldsymbol{\mu}_q^{(0)})
\sigma_q  B_q +
\end{equation} 
$$
  B_{q}  
\sigma_{q-1} B_{q-1} +
B_{q-1} \sigma(\boldsymbol{\mu}_q^{(1)}) \sigma^*_{q-1} B_{q}  -  
B_{q-1} \sigma(\boldsymbol{\mu}^{(1)}_{q})   \sigma^*_{q-1} 
\sigma_{q-1} B_{q-1}. 
$$
Then, properties \eqref{eq.Bi} of matrices $B_j$ and formulae \eqref{eq.symb.P}, 
\eqref{eq.Laplace.mu.symb.rel.02} imply  
\begin{equation*} 
 B_q \delta_{q,\boldsymbol{\mu}} B_q \Big(\mathcal{N}^{(q)}_\sigma \!+\! \tilde
 \sigma({\mathcal M}_{q-1} (A))\Big) \!=\!
 B_q \sigma^*_{q} \sigma(\boldsymbol{\mu}_q^{(0)}) \sigma_q  B_q \!+ \!
 B_q \sigma_{q-1} \sigma(\boldsymbol{\mu}_q^{(1)})  \sigma^*_{q-1}  \sigma_{q-1} B_{q-1} ,
\end{equation*}
\begin{equation*} 
\tilde \sigma({\mathcal M}_{q} (A)) \tilde \sigma({\mathcal M}_{q-1} (A)) = 
\tilde \sigma({\mathcal M}_{q-1} (A))\tilde \sigma({\mathcal M}_{q-1} (A)),
\end{equation*}
\begin{equation*} 
\tilde \sigma({\mathcal M}_{q} (A)) \mathcal{N}^{(q)}_\sigma = 
 B_q\sigma_{q-1} \sigma(\boldsymbol{\mu}_q^{(1)}) \sigma^*_{q-1}  B_q + 
B_{q-1} \sigma^*_{q-1} 
 \Big(\delta^{-1}_{q,\boldsymbol{\mu}} \sigma^*_{q} \sigma(\boldsymbol{\mu}_q^{(0)})
\sigma_q \Big) B_q  +
\end{equation*}
\begin{equation*}
B_{q-1} \sigma_{q-1}^* \sigma_{q-1} B_{q-1}  -B_{q}  
\sigma_{q-1} \sigma(\boldsymbol{\mu}_q^{(1)} ) \sigma^*_{q-1}  \sigma_{q-1} B_{q-1} +
\end{equation*}
\begin{equation*}
B_{q-2} \sigma^*_{q-2}
 \sigma(\boldsymbol{\mu}_q^{(1)}) \sigma^*_{q-1} B_{q}  -  
B_{q-2} \sigma^*_{q-2} \sigma(\boldsymbol{\mu}^{(1)}_{q})   \sigma^*_{q-1} 
\sigma_{q-1} B_{q-1} =
\end{equation*}
 \begin{equation*}
B_q \sigma_{q-1} \sigma(\boldsymbol{\mu}_q^{(1)}) 
\sigma^*_{q-1}  B_q + B^q_{q-1} \sigma_{q-1}^* \sigma_{q-1} B_{q-1}
-B_{q}  
\sigma_{q-1} \sigma(\boldsymbol{\mu}_q^{(1)})  \sigma^*_{q-1}  \sigma_{q-1} B_{q-1}.
\end{equation*}
Thus, we arrive at the following identity:
\begin{equation} \label{eq.Sq.Nq}
\sigma_{\vec{s},\vec{t}} (S_{q,1} (A,{\mathfrak D}_{\boldsymbol{\mu}})) \Big( \mathcal{N}^{(q)}_\sigma 
+ \tilde \sigma({\mathcal M}_{q-1} (A))\Big) = B_q \delta_{q,\boldsymbol{\mu}} B_q + 
\end{equation}
$$
\sum_{j=0}^{q-1} B_j \delta_j B_j +
B_{q-2} \sigma^*_{q-2}
 \sigma(\boldsymbol{\mu}_q^{(1)}) \sigma^*_{q-1} B_{q}  -  
B_{q-2} \sigma^*_{q-2} \sigma(\boldsymbol{\mu}^{(1)}_{q})   \sigma^*_{q-1} 
\sigma_{q-1} B_{q-1}.
$$
Now, if $U = ( u_1, \dots u_q)$ satisfies 
$$
\sigma_{\vec{s},\vec{t}} (S_{q,1} (A,{\mathfrak D}_{\boldsymbol{\mu}})) \Big( \mathcal{N}^{(q)}_\sigma 
+ \tilde \sigma({\mathcal M}_{q-1} (A))\Big) U= 0, 
$$
then 
$$
\delta_{q,\boldsymbol{\mu}} u_q = 0, \,  \delta_{q-1} u_{q-1} =0, \, \delta_j u_j =0, \, 0\leq j\leq q-3, 
$$
$$
\sigma^*_{q-2}
 \sigma(\boldsymbol{\mu}_q^{(1)}) \sigma^*_{q-1} u_{q}  -  
 \sigma^*_{q-2} \sigma(\boldsymbol{\mu}^{(1)}_{q})   \sigma^*_{q-1} 
\sigma_{q-1} u_{q-1} + \delta_{q-2} u_{q-2}=0.
$$
Immediately we see that $u_j=0$ for all $0\leq j \leq q$, $j\ne q-2$, because symbolic matrices 
$\delta_{q,\boldsymbol{\mu}}$, 
$\delta_{j}$ are invertible for all $(x,\zeta) \in T^{\ast} X $ with $\zeta \ne 0$. Therefore  
$\delta_{q-2} u_{q-2}=0$ and then, again $u_{q-2}=0$ for the same reason.
Hence the matrices 
$$
\sigma_{\vec{s},\vec{t}} (S_{q,1} (A,{\mathfrak D}_{\boldsymbol{\mu}})) \Big( \mathcal{N}^{(q)}_\sigma 
+ \tilde \sigma({\mathcal M}_{q-1} (A))\Big) \mbox{ and } 
\sigma_{\vec{s},\vec{t}} (S_{q,1} (A,{\mathfrak D}_{\boldsymbol{\mu}}))
$$ 
are 
invertible for all $(x,\zeta) \in T^{\ast} X $ with $\zeta \ne 0$, too, i.e. the operator 
$S_{q,1} (A, \mathfrak{D}_{\boldsymbol{\mu}})$ is Douglis-Nirenberg elliptic. 
\end{proof}

\subsection{Maxwell's and Stokes' type operators for induced elliptic complexes} 
For $N\geq 1$ introduce non-steady Maxwell's type operators 
$$
{\mathcal M}^{(i)}_{q,\boldsymbol{\mu}} (A,\mathbf{b}\partial _t) = 
\sum_{j=0}^q B_j b_j  B_j  \partial _t 
 + {\mathcal M}^{(i)}_{q,\boldsymbol{\mu}} (A), i=1,2,
$$
with a vector ${\mathbf b} = (b_1, \dots b_q)$ consisting of complex entries.

\begin{lem}
Let  the coefficients of the operators $A_j$, $0\leq j \leq N-1$, do not depend on the time 
variable $t$. 
If \eqref{eq.coh} and  \eqref{eq.commute.mu} are fulfilled for all $0\leq j \leq q$ then 
for any real vector  ${\mathbf b}$ we  have
\begin{equation*} 
 {\mathcal M}^{(1)}_{q,\boldsymbol{\mu}} (A,-\mathbf{b}\partial _t) {\mathcal M}^{(0)}_{q,\boldsymbol{\mu}} 
(A,\mathbf{b}\partial _t)= 
\end{equation*}
$$
B_{q} (A_{q-1} \boldsymbol{\mu}^{(1)}_{q} A^*_{q-1} - b_q^2 \partial^2_t) 
B_{q}+
\sum_{j=1}^{q-1}  B_{j}  \Big( \Delta_{j,\boldsymbol{\mu}}  -b_j^2 \partial^2_t  \Big) B_{j} +
B_{0}  (\Delta_{0,\boldsymbol{\mu}} -  b_0^2\partial^2_t )B_{0},
$$
$1\leq q\leq N$.  In particular, Maxwell's type operators ${\mathcal M}^{(j)}
_{N,\boldsymbol{\mu}} (A,\pm \mathbf{b}\partial _t)$ are elliptic for $N\geq 1$, if $\mathbf{b}\in 
{\mathbb R}^q$, $m_j=1$, $|b_j|>0$,  and Assumption \ref{asm.mu} is fulfilled for all $0\leq j \leq N$.
\end{lem}

Next, for $\mathbf{b}\in {\mathbb C}^q$ we set 
$$
S_{q,a} (A, \mathbf{b}{\mathcal L}_{\boldsymbol{\mu}}) =  \sum_{j=0}^q B_j b^2_j 
\big(\partial _t + {\mathfrak D}_{j, \boldsymbol{\mu}} \big)B_j  +  a {\mathcal M}_{q} (A),
$$
representing models with the leading 'parabolic' part, and 
$$
S_{q,a} (A, \mathbf{b} {\mathcal H}_{\boldsymbol{\mu}}) =  \sum_{j=0}^q B_j b_j 
\big(\partial^2 _t + {\mathfrak D}_{j, \boldsymbol{\mu}} \big)B_j + a \, {\mathcal M}_{q} (A)
$$
corresponding to models with the leading 'hyperbolic' part (as before, $a=a_q $ equals to $1$ or $0$). 
It is worth to note that, similarly 
to steady case, the main part of the Stokes operator 
$S_{N,0} (A,{\mathfrak D}_{\boldsymbol{\mu}}, \mathbf{b}
\partial^2 _t)$ could be easily factorized if one chooses 
suitable operators ${\mathfrak D}_{j,\boldsymbol{\mu}}$ and numbers $b_j$.
 
\begin{lem} 
\label{l.factor.t}
Let  the coefficients of the operators $A_j$, $0\leq j \leq N-1$, do not depend on the time 
variable $t$. 
If for all $0\leq j \leq q$ identities  \eqref{eq.coh} and \eqref{eq.commute.mu} are fulfilled  
then for any real vector $\mathbf{b} \in {\mathbb R}^q$ we have
\begin{equation*} 
{\mathcal M}^{(1)}_{q,\boldsymbol{\mu}} (A,-\iota \mathbf{b}\partial _t) 
{\mathcal M}^{(0)}_{q,\boldsymbol{\mu}} (A,\iota \mathbf{b}\partial _t)
= 
\end{equation*}
$$
B_{q} b_q^2 \big (\partial_t ^2 + A_{q-1} 
\boldsymbol{\mu}_{q-1} ^{(1)} A^*_{q-1} \big) B_{q} + 
\sum_{j=0}^{q-1} B_j b_j^2 (\partial_t ^2 +\Delta_{j,\boldsymbol{\mu}_j}) B_j.
$$
In particular, if ${\mathfrak D}_{j,\boldsymbol{\mu}}= \Delta_{j,\boldsymbol{\mu}}$ for all 
$0\leq j \leq q $, then 
\begin{equation*} 
S_{N,0} (A, {\mathbf b}{\mathcal H}_{\boldsymbol{\mu}}) = 
{\mathcal M}^{(1)}_{q,\boldsymbol{\mu}} (A,-\iota \mathbf{b}\partial _t) 
{\mathcal M}^{(0)}_{q,\boldsymbol{\mu}} (A,\iota \mathbf{b}\partial _t). 
\end{equation*}
\end{lem}
Thus, 
for the first order complex \eqref{eq.complex} 
one may treat 
operator ${\mathcal M}^{(0)}_{q,\boldsymbol{\mu}} (A,\iota \mathbf{b}\partial _t)$ 
as the first order 'wave operator'. 

\subsection{Parametrices for steady Maxwell' and Stokes' operators.} 
At this point we note that the (both Petrovskii and Douglis-Nirenberg) ellipticity implies the regularity property 
for solutions to the operators ${\mathcal M}_{N,\boldsymbol{\mu}}(A)$, $S_{q,1}(A,{\mathfrak D}_{\boldsymbol{\mu}})$ 
and existence of parametrices (and even fundamental solutions) for them, see, for instance, \cite{AV}, 
\cite[\S 2.3]{Tark35},  \cite[\S 2.2.9, \S 4.4]{Tark36},  \cite[Ch. 2]{EgShu}, \cite[Theorem 8.69]{WRL}. This 
actually results in many useful integral formulae for solutions to related systems of differential equations,
see for instance,  \cite[\S 2.4, \S 2.5]{Tark35}.

Let us indicate a way to construct parametrices for Maxwell's and Stokes'  operators using suitable 
kernels for  the generalized Laplacians $\Delta_{j,\boldsymbol{\mu}}$ of elliptic complex 
\eqref{eq.complex}. Namely, if $\Delta_{j,\boldsymbol{\mu}}$ are strongly elliptic operators, each of them admits a 
parametrix, say, $\Phi_{j,\boldsymbol{\mu}}$, i.e. such a pseudo-differential operator on $X^0$ that on 
$C^\infty_0 (X,E_j)$ we have 
\begin{equation} \label{eq.param.Delta}
\Phi_{j,\boldsymbol{\mu}} \Delta_{j,\boldsymbol{\mu}} + \Pi^{L}_{j,\boldsymbol{\mu}} = I, \,\, 
\Delta_{j,\boldsymbol{\mu}}  \Phi_{j,\boldsymbol{\mu}} + \Pi^R_{j,\boldsymbol{\mu}} = I , 
\end{equation}
with pseudo-differential operators $\Pi^R_{j,\boldsymbol{\mu}}$, $\Pi^L_{j,\boldsymbol{\mu}}$ of negative orders, 
where $0\leq j \leq N$ and $I$ is the identity operator; in some situation one needs smoothing  operators 
$\Pi^R_{j,\boldsymbol{\mu}}$, $\Pi^L_{j,\boldsymbol{\mu}}$, i.e. 
the pseudo-differential operators of order minus infinity. If $ \Pi^L_{j,\boldsymbol{\mu}} =0$ in 
\eqref{eq.param.Delta} then 
$\Phi_{j,\boldsymbol{\mu}}$ is a left  fundamental solution for $\Delta_{j,\boldsymbol{\mu}}$ on $X$;
similarly, if $ \Pi^R_{j,\boldsymbol{\mu}} =0$ then $\Phi_j$ is a right  
fundamental solution for $\Delta_{j,\boldsymbol{\mu}}$ on $X$. 
In particular, if $\Delta_{j,\boldsymbol{\mu}}$ satisfy the so-called 
Uniqueness Condition in small on $X$ then 
$\Pi^L_{j,\boldsymbol{\mu}}=\Pi^R_{j,\boldsymbol{\mu}_j} =0$, 
i.e. $\Phi_{j,\boldsymbol{\mu}}$ is the bilateral fundamental solution for $\Delta_{j,\boldsymbol{\mu}_j}$ on $X$, 
see, for instance \cite[\S 4.4]{Tark36}.

\begin{thm} \label{t.right.M} 
Let complex \eqref{eq.complex} be elliptic, Assumption {\rm \ref{asm.mu}} be fulfilled,  
\eqref{eq.coh.symb} and $m_j + \tilde m_j = m_{j-1} + \hat m_j=m$ 
for all $0\leq j \leq N$. Then the operator 
$$ {\mathcal F}^{(1)}_{N,\boldsymbol{\mu}} (A)  =
{\mathcal M}^{(0)}_{N,\boldsymbol{\mu}} (A) \Big(
\sum_{j=0}^{N-1} B_j \Phi_{j,\boldsymbol{\mu}} B_j \Big),
$$ 
is a parametrix for ${\mathcal M}^{(1)}_{N,\boldsymbol{\mu}} (A)$. Moreover, 
if $\Phi_{j,\boldsymbol{\mu}}$,  $0\leq j \leq N$, are right fundamental solutions for 
$\Delta_{j,\boldsymbol{\mu}}$ then $ {\mathcal F}^{(1)}_{N,\boldsymbol{\mu}} (A)$ is a right fundamental solution to 
${\mathcal M}^{(1)}_{N,\boldsymbol{\mu}} (A)$. 
\end{thm}

\begin{proof} It follows from \eqref{eq.P.mu} that the operator $ {\mathcal F}^{(0)}_{N,\boldsymbol{\mu}} 
(A)$ is a right parametrix for 
${\mathcal M}^{(1)}_{N,\boldsymbol{\mu}} (A)$ if $\Phi_{j,\boldsymbol{\mu}}$ are parametrices 
for $\Delta_{j,\boldsymbol{\mu}}$, respectively (similarly, a right fundamental solution if 
$ \Pi^R_{j,\boldsymbol{\mu}}=0$), $0\leq j \leq N$. Finally, we note that for 'elliptic' operators a right 
parametrix is a left parametrix, too, see, for instance, \cite[\S 2.2.9]{Tark36}.
\end{proof}

Similarly, we obtain the following statement.

\begin{thm} \label{t.left.M} 
Let complex \eqref{eq.complex} be elliptic, Assumption {\rm \ref{asm.mu}} be fulfilled,  
\eqref{eq.coh.symb} be true and $m_j + \tilde m_j = m_{j-1} + \hat m_j=m$ 
for all $0\leq j \leq N$. Then the operator 
$$ 
{\mathcal F}^{(0)}_{N,\boldsymbol{\mu}} (A) = \Big(
\sum_{j=0}^{N-1} B_j \Phi_{j,\boldsymbol{\mu}} B_j \Big) {\mathcal M}^{(1)}_{N,\boldsymbol{\mu}} (A) 
$$ 
is a parametrix for ${\mathcal M}^{(0)}_{N,\boldsymbol{\mu}} (A)$.  Moreover, 
if $\Phi_{j,\boldsymbol{\mu}}$,  $0\leq j \leq N$, are left fundamental solutions for 
$\Delta_{j,\boldsymbol{\mu}}$ then $ {\mathcal F}^{(0)}_{N,\boldsymbol{\mu}} (A)$ is a left fundamental solution to 
${\mathcal M}^{(0)}_{N,\boldsymbol{\mu}} (A)$. 
\end{thm}

Finally, let us write down a fundamental solution to Stokes' operator $S_{q,1} (A,\mathfrak{D}_{\boldsymbol{\mu}})$ 
in a particular case. With this purpose, let $\mathcal{N}^{(q)} (A)$ be given by
\begin{equation} \label{eq.Nq.op}
B_{q} \Phi_{q,\boldsymbol{\mu}} A^*_{q}\boldsymbol{\mu}_q^{(0)} A_q  B_q \!+\!
  B_{q}  A_{q-1} B_{q-1} \!+\! B_{q-1} \boldsymbol{\mu}_q^{(1)} A^*_{q-1} B_{q}  \! - \! 
B_{q-1} \boldsymbol{\mu}^{(1)}_{q}  A^*_{q-1} A_{q-1} B_{q-1} .
\end{equation}

\begin{thm} \label{t.stokes.fund.sol} 
Let complex \eqref{eq.complex} be elliptic, \eqref{eq.coh} be true, $N\geq 2$ and $1\leq q \leq N-1$. 
Let also $m_j = m$ and $\boldsymbol{\mu}^{(i)}_j = 0$ for all $i=1,2$ and all $0\leq j \leq q-1$, 
 $m_q+\tilde m_q =m$,  Assumption {\rm \ref{asm.mu}} be fulfilled for $j=q$ and 
\begin{equation} \label{eq.mu.mu} 
A^*_{q-2} \boldsymbol{\mu}^{(1)}_{q}  A^*_{q-1} =0.\
\end{equation}  
If $\Phi_{j}$  are 
right fundamental solutions for $\Delta_j$, $0\leq j\leq q-1$, and 
$\Phi_{q,\boldsymbol{\mu}}$ is a 
bilateral fundamental solution to $\Delta_{q,\boldsymbol{\mu}}$, then the operator 
$$
\mathfrak{F}_{q,\boldsymbol{\mu}} (A)= \Big( 
\mathcal{N}^{(q)} (A) + {\mathcal M}_{q-1} (A)\Big) \Big( B_q \Phi_{q,\boldsymbol{\mu}} B_q + 
\sum_{j=0}^{q-1} B_j \Phi_j B_j\Big) 
$$ 
is a right fundamental solution to $S_{q,1} (A,\Delta_{\boldsymbol{\mu}})$. 
\end{thm}

\begin{proof} Indeed, as $A_{i+1} \circ A_i \equiv 0$ then  we have 
\begin{equation} \label{eq.complex.rel}
A^*_{i} \circ A^*_{i+1} \equiv 0, \, \Delta_{i+1} A_i = A_i \Delta_{i} = A_i A_i^* A_i, \, A^*_i\Delta_{i+1}  =  \Delta_{i} A_i^* =  A_i^* A_i A^*_i.
\end{equation}
Next, using \eqref{eq.complex.rel} we conclude that 
\begin{equation*} 
\Delta_{j, \boldsymbol{\mu}} A^*_j \boldsymbol{\mu}^{(0)}_{j} A_j=
A^*_j \boldsymbol{\mu}^{(0)}_{j} A_j  A^*_j 
\boldsymbol{\mu}^{(0)}_{j} A_j =  
A^*_j \boldsymbol{\mu}^{(0)}_{j} A_j \Delta_{j, \boldsymbol{\mu}}
\end{equation*}
and, if $\Phi_{j, \boldsymbol{\mu}} $ is a bilateral fundamental solution
for $\Delta_{j, \boldsymbol{\mu}}$, then 
\begin{equation} \label{eq.Laplace.mu.rel.02}
 A^*_j \boldsymbol{\mu}^{(0)}_{j} A_j \Phi_{j, \boldsymbol{\mu}}= 
 \Phi_{j, \boldsymbol{\mu}} A^*_j \boldsymbol{\mu}^{(0)}_{j} A_j .
\end{equation}
Hence calculating as in the proof of Proposition \ref{t.stokes.ell} (see formulae 
\eqref{eq.Nq}, \eqref{eq.Sq.Nq}) and applying \eqref{eq.mu.mu}, we obtain
\begin{equation} \label{eq.Sq.Nq.op}
S_{q,1} (A,\Delta_{\boldsymbol{\mu}}) \Big( \mathcal{N}^{(q)} (A) + {\mathcal M}_{q-1} (A)\Big) 
= B_q \Delta_{q,\boldsymbol{\mu}} B_q + \sum_{j=0}^{q-1} B_j \Delta_j B_j ,
\end{equation}
i.e. $\mathfrak{F}_{q,\boldsymbol{\mu}} (A)$ is a right fundamental solution for 
$S_{q,1} (A,\Delta_{\boldsymbol{\mu}})$. 
\end{proof}

Note that for the de Rham complex the pseudo-differential operator  
$$
{\mathfrak F}_{1,\mu} (d) =
\left( \begin{array}{cccccccc} 
\mu^{-1}  
 \varphi \, {\rm rot} \, {\rm rot} \, \varphi   &  -\varphi \, \nabla \\ 
{\rm div} \, \varphi   & - \mu  \, I \\
 \end{array}
\right) 
$$
is closely related  to the so-called steady Ozeen tensor for Stokes' system 
$S_{q,1} (d,\Delta_{\boldsymbol{\mu}})$ in $3D$-Hydrodynamics, see \cite{Oz}, 
where  $\varphi$ is the standard fundamental solution of the Laplace operator in ${\mathbb R}^3$, 
$\boldsymbol{\mu}^{(0)}_1 = \mu I_n $, $\boldsymbol{\mu}^{(1)}_1 = \mu $, $\mu>0$. 

\subsection{Parametrices for evolutionary  operators}
Under very mild assumptions on the coefficients of the operator 
${\mathcal L}_{j,\boldsymbol{\mu}}$, $0\leq j \leq N$,  it admits the (unique) fundamental solution 
$\Psi_{j,\boldsymbol{\mu}}$ on $X \times [0,T]$ solving the Cauchy problem for ${\mathcal L}_{j,\boldsymbol{\mu}}$, 
with initial data on the plane $t=0$, see, for instance, \cite{eid}, \cite[Ch 1, \S 7 and Ch. 9]{Frid}. 

Again, let us write down a fundamental solution to Stokes' operator 
$S_{q,1} (A,{\mathbf b}{\mathcal L}_{\boldsymbol{\mu}})$ 
in a particular case. With this purpose  let $\mathcal{N}^{(q)} (A,t)$ be given by
\begin{equation*} 
B_{q} 
\Psi_{q,\boldsymbol{\mu}} A^*_{q}\boldsymbol{\mu}_q^{(0)} \! A_q B_q \!+\!
 B_{q}  A_{q-1} B_{q-1} \!+\!B_{q-1} \boldsymbol{\mu}_q^{(1)} A^*_{q-1} B_{q}  \! - \! 
B_{q-1} (\boldsymbol{\mu}^{(1)}_{q} \! A^*_{q-1} A_{q-1} \!+\!\partial_t)B_{q-1} .
\end{equation*}

\begin{thm} \label{t.stokes.fund.sol.t} 
Let complex \eqref{eq.complex} be elliptic, \eqref{eq.coh} be true, $N\geq 2$ and $1\leq q \leq N-1$, and the 
coefficients of the operators $A_j$ do not depend on $t$. Let also $\mathbf{b}=(0,0,\dots, 0, 1)$, $m_j = m$ 
and $\boldsymbol{\mu}^{(i)}_j = 0$ for all $i=1,2$ for all $0\leq j \leq q-1$,  
$m_q+\tilde m_q =m$,  Assumption {\rm \ref{asm.mu}} be fulfilled for $j=q$
and \eqref{eq.mu.mu} hold true. If $\Phi_{j}$  are 
right fundamental solutions for $\Delta_j$, $0\leq j\leq q-1$, and 
$\Phi_{q,\boldsymbol{\mu}}$, $\Psi_{q,\boldsymbol{\mu}}$ are 
bilateral fundamental solution to $\Delta_{q,\boldsymbol{\mu}}$, $(\partial_t + \Delta_{q,\boldsymbol{\mu}})$, 
respectively, then the operator 
$$
\mathfrak{F}_{q,\boldsymbol{\mu}} (A,t )= \Big( 
\mathcal{N}^{(q)} (A,t) + {\mathcal M}_{q-1} (A)\Big)\Big( B_q \Phi_{q,\boldsymbol{\mu}} B_q + 
\sum_{j=0}^{q-1} B_j \Phi_j B_j\Big) 
$$ 
is a right fundamental solution to  $S_{q,1} (A,\mathbf{b} (\partial _t +
\Delta_{\boldsymbol{\mu}}))$. 
\end{thm}

\begin{proof}
Indeed, similarly to \eqref{eq.Laplace.mu.rel.02}, if the coefficients of the operators $A_j$ do not depend on $t$ 
and $\Psi_{j, \boldsymbol{\mu}} $ is a bilateral fundamental solution to $(\partial_t +\Delta_{j, 
\boldsymbol{\mu}})$, then 
\begin{equation} \label{eq.Laplace.mu.rel.02.t}
 A^*_j \boldsymbol{\mu}^{(0)}_{j} A_j \Psi_{j, \boldsymbol{\mu}}= 
 \Psi_{j, \boldsymbol{\mu}} A^*_j \boldsymbol{\mu}^{(0)}_{j} A_j .
\end{equation}
Hence calculating as in the proof of Theorem \ref{t.stokes.fund.sol}  (see formulae 
\eqref{eq.Nq.op}, \eqref{eq.Sq.Nq.op}) and applying \eqref{eq.mu.mu} and \eqref{eq.Laplace.mu.rel.02.t}, we obtain
\begin{equation*} 
S_{q,1} (A,\partial_t + \Delta_{\boldsymbol{\mu}}) \Big( \mathcal{N}^{(q)} (A,t)
+ {\mathcal M}_{q-1} (A)\Big) = B_q \Delta_{q,\boldsymbol{\mu}} B_q + \sum_{j=0}^{q-1} B_j 
\Delta_j B_j ,
\end{equation*}
i.e. $\mathfrak{F}_{q,\boldsymbol{\mu}} (A,t)$ is a right fundamental solution for 
$S_{q,1} (A,\partial_t + \Delta_{\boldsymbol{\mu}})$. 
\end{proof}

\subsection{Natural perturbations.} 
In order to define suitable perturbations of Stokes' type systems ${\tilde S}_q (A)$
(and ${\hat S}_q (A)$) related to the Navier-Stokes equations, 
\cite{PaSh} proposed to introduce two  bilinear mappings ${\mathcal Q}_{q,j}$, satisfying 
\begin{equation} \label{eq.nonlinear.M}
{\mathcal Q}_{q,1,x}: E_{q+1,x} \otimes E_{q,x}  \to E_{q,x}, \,\,
{\mathcal Q}_{q,2,x}: E_{q,x} \otimes E_{q,x}\to E_{q-1,x}, 
\end{equation} 
at each point $x\in X$. Then one may set for a sufficiently differentiable section $v$ of the vector bundle $E_q$:
\begin{equation} \label{eq.nonlinear}
{\mathcal N}_q (v) (x) = {\mathcal Q}_{q,1,x} ((A_{q} v) (x),v(x)) + A_{q-1} 
{\mathcal Q}_{q,2,x} (v(x),v(x))  .
\end{equation}  
For the de Rham complex at the degree $q=1$, corresponding to 
the classical Navier-Stokes equations, this  leads to the so-called Lamb form 
(see \cite[\S 15]{LaLi} for $n=3$) of the related non-linear 
part: 
$$
{\mathcal N}_1 (v)  = \star (\star d_1 v \wedge v) + 
d_0 \star(v \wedge \star v)/2= ( (\mathrm{curl} \, \vec v) \times \vec v + (1/2)\nabla |\vec v|^2\mbox{ for } n=3),
$$
where $\star: \Lambda^q \to \Lambda^{n-q} $ is the Hodge $\star$-operator, $\wedge$ is the 
exterior product of the differential forms (and $\vec{c} \times \vec{d}$ stands for the vector product of 
$3$-vectors $\vec{c}$ and $\vec{d}$). 

Apparently Stokes' type systems $S_{q,a} (A, {\mathfrak D}_{\boldsymbol{\mu}})$ 
have a more complicated  structure. In the simplest case where all the orders $m_j$ are the same 
for all $0\leq j \leq q$ we suggest to introduce multi-linear differential operators of order $(m-1)$:
$$
Q_{q,1}: (\oplus_{i=1}^{q+1} E_i) \otimes (\oplus_{i=0}^{q} E_i) \to \oplus_{i=0}^{q} E_i, \, 
Q_{q,2}:  (\oplus_{i=0}^{q} E_i) \otimes (\oplus_{i=0}^{q} E_i)  \to \oplus_{i=0}^{q-1} E_i, 
$$
and to set the following natural non-linear perturbations
$$
{\mathcal N}_q (U_q) = Q_{q,1} (A^\cdot U_q,U_q)
+ A^\cdot Q_{q,2} ( U_q,  U_q)  ,
$$
of the differential operator $\sum_{j=0}^q B_j {\mathfrak D}_{j,\boldsymbol{\mu}} B_j$ acting on sections 
$U_q = (u_0, \dots, u_q) $ 
of the bundle ${\mathfrak E}_q = \oplus_{j=0}^q E_j$, $0\leq q \leq N$. 

\section{Maxwell's and Stokes' type operators for the de Rham complex}
\label{s.MS.de Rham}

\subsection{The de Rham complex} 
Let $T^*_{\mathbb C} X$ be the complexified tangential bundle of $X$ and 
let $\Lambda^q= \Lambda^q T^*_{\mathbb C} X$  be the bundle of complex valued 
exterior differential forms  of degree $q$, $0\leq q \leq n$, over $X$. In each coordinate 
neighbourhood $\mathcal O$ on $X$ any differential form $u$ admits local representation 
$$
u_{|{\mathcal O}}(x) = \sum_{\# I=q} u _I (x) dx_I  
$$ 
where, for $\# I=q$, we consider  $I = (i_1, \dots i_q)$ as 
a multi-index with $1\leq i_1 < \dots <i_q \leq n$, 
$ dx_I = dx_{i_1}\wedge \dots \wedge  dx_{n}$, 
$\{ dx_j\}_{j=1}^n$ is a basis in $T^*_{\mathbb C} X$ and $\wedge$ is 
the exterior product of differentials satisfying 
\begin{equation} \label{eq.wedge}
dx_i \wedge dx_j = - dx_j \wedge dx_i.
\end{equation}
Then the exterior differential operator $d_q$ is defined by local representations 
$$
d_q u_{|{\mathcal O}}(x) = \sum_{i=1}^n\sum_{\# I=q} 
\partial_i u _I (x) dx_i \wedge dx_I  .
$$ 
Using \eqref{eq.wedge} we easily conclude that  
\begin{equation} \label{eq.deRham.base.rel} 
d_{q+1} \circ d_q =0 
\end{equation}
 and then 
we obtain the de Rham complex 
\begin{equation} \label{eq.deRham}
0\rightarrow 
C^\infty(X,\Lambda^0)\stackrel{d_0}{\rightarrow}C^\infty(X, \Lambda^1)\stackrel{d_1}{
\rightarrow}C^\infty(X, \Lambda^2) \rightarrow \dots 
\stackrel{d_{n-1}}{\rightarrow} C^\infty(X, \Lambda^n) \rightarrow 0, 
\end{equation}
of exterior differentials on the differential forms, see, for instance,
\cite[Ch 3, \S 2.5]{Car}, \cite{deRham}, \cite[\S 1.2.6]{Tark35}. 
Of course, for $X={\mathbb R}^n$ the bundle $\Lambda^q$ may be identified 
with ${\mathbb R}^n \times {\mathbb C}^{k_q}$ with 
$k_q =\Big(\begin{array}{ll} n \\ q
\end{array}\Big)$ 
and its sections can be treated as vector-columns of functions with $k_q$ components.
Then the de Rham complex is the Hilbert compatibility 
complex degenerated by the gradient 
operator $ \nabla = d_0 $. In this case, 
in addition to  \eqref{eq.deRham.base.rel},  we also have 
$$
\Delta_q = d_q ^{*} d_q + d_{q-1} d_{q-1} ^{*} = - \Delta I_{k_q}, \,\, 0\leq q\leq n,
$$
where $\Delta$ is the usual Laplace operator in ${\mathbb R}^n$.  

For three dimensional space, more 
familiar within classical Physics, we may interprete the de Rham complex as follows:
\begin{equation} \label{eq.deRham.n=3} 
d_0 = \nabla, \, 
d_1 = \mathrm{curl} = \left(
\begin{array}{ccc}
0 & - \partial_3 
&   \partial_2 
\\
  \partial_3 & 0 &   -\partial_1  
\\
-\partial_2 &  \partial_1 
& 0  \\
\end{array}
\right), \,  
d_2 = \mathrm{div} = \Big(\partial_1 
, \partial_2
, \partial_3 
\Big) . 
\end{equation}
However we still may define the compatibility de Rham complex in ${\mathbb R}^3$ with the use
of the classical algebraic constructions: 
$$
d_1 u =\nabla \times  \vec{u}, \, d_2 v = \nabla \cdot \vec{v}  
$$ 
for vector fields $\vec{u}$, $\vec{v}$ (here $\vec{c} \cdot \vec{d}$ means 
the inner product of vectors $\vec{c}$ and $\vec{d}$). In the higher dimensions the standard 
algebraic constructions do not work in general.  

As $N=n$, Stokes's operator $S_q (d,\Delta)$ over the field ${\mathbb R}$ has dimension 
$ r_q = \sum_{j=0}^q \Big(\begin{array}{ll} n \\ j \end{array}\Big)$; 
in particular, $r_0=1$, $r_1 = 1 +n$, $ r_{n-1} = 2^n-1$, $ r_n = 2^n$. 
Of course, the dimension $r_{q, {\mathbb C}}$ over the field ${\mathbb C}$ equals to $2 r_{q}$;
in particular, $ r_{n, {\mathbb C}} = 2^{n+1}$.

\subsection{Maxwell's and Stokes' type systems in ${\mathbb R}^3$}

For the de Rham complex over ${\mathbb  R}^3$ we have: $r_0=1$, $r_1 = 4$, $ r_{2} = 7$, $ r_3 = 8$; 
in particular, $ r_{n, {\mathbb C}} = 16$. Thus, our method echoes with 
 the sedeonic approach by V.L. Mironov and S.V. Mironov, see \cite{Mir15}, 
\cite{Mir18}, \cite{Mir20A}, \cite{Mir20B}, 
for compact symmetric formulations of the Laws of classical Physics in ${\mathbb R}^3$.
 
The typical form of Stokes' system $S_0 (d,{\mathcal L}_{\boldsymbol{\mu}})$ is the following: 
$$
S_0 (d,\mathcal{L}_{\boldsymbol{\mu}}) = \partial_t +\mathfrak{D}_{0,\boldsymbol{\mu}} =  \partial_t - \mathrm{div} \boldsymbol{\mu}_0^{(0)} \nabla + \vec{a}_0 \cdot \nabla + M_0, 
$$
with a self-adjoint functional matrix $\boldsymbol{\mu}_0^{(0)}$, a scalar function $M_0$ and 
a functional vector $\vec{a}_0$, see \eqref{eq.diagonal}, \eqref{eq.diagonal.low}, that 
is a standard second order equation of the Mathematical Physics. 

Next, according to \eqref{eq.diagonal}, \eqref{eq.diagonal.low}, we have 
$$
\mathfrak{D}_{1,\boldsymbol{\mu}} =  \mathrm{curl} \, \boldsymbol{\mu}_1^{(0)} \, 
\mathrm{curl}  - 
\nabla \, \boldsymbol{\mu}_1^{(1)} \, \mathrm{div}  + C_1 \, \mathrm{curl} + \vec{a}_1 \, \mathrm{div} + M_1
$$
with self-adjoint functional $(3\times 3)$ matrices $\boldsymbol{\mu}_1^{(0)}$, $M_1$,  a functional $(3\times 3)$-
matrix $C_1$, a functional $3$-vector $\vec{a}$ and a scalar function $\boldsymbol{\mu}_1^{(1)}$.  
Then the typical steady  Stokes' type system 
$ S_1 (d, \mathfrak{D}_{\boldsymbol{\mu}})$ is given by 
\begin{equation} \label{eq.S1.deRham}
 S_{1,1} (d,\mathfrak{D}_{\boldsymbol{\mu}})= \left( \begin{array}{ccc} 
\mathfrak{D}_{1,\boldsymbol{\mu}} & \nabla  \\ 
-\mathrm{div} & \mathfrak{D}_{0,\boldsymbol{\mu}}  \\ 
\end{array}
\right). 
\end{equation}
Similarly, at the second step of the de Rham complex over ${\mathbb  R}^3$ we have:
$$
\mathfrak{D}_{2,\boldsymbol{\mu}} =  - 
\nabla \, \boldsymbol{\mu}_2^{(0)} \, \mathrm{div} + 
\mathrm{curl} \, \boldsymbol{\mu}_2^{(1)} \, \mathrm{curl}  + C_2 \, \mathrm{curl} + \vec{a}_2 \, \mathrm{div} + M_2
$$
with self-adjoint functional $(3\times 3)$ matrices $\boldsymbol{\mu}_2^{(1)}$, $M_2$,  a functional $(3\times 3)$-matrix $C_2$, a functional $3$-vector $\vec{a}_2$ and a scalar function $\boldsymbol{\mu}_2^{(0)}$. Consequently, 
$$
 S_{2,1} (d,\mathfrak{D}_{\boldsymbol{\mu}})= \left( \begin{array}{ccc} 
\mathfrak{D}_{2,\boldsymbol{\mu}} & {\rm curl} & 0 \\ 
{\rm curl} & \mathfrak{D}_{1,\boldsymbol{\mu}}  & \nabla \\ 
0 & - {\rm div} & \mathfrak{D}_{0,\boldsymbol{\mu}} \\
\end{array}
\right). 
$$
Finally, at the third step we have
$$
\mathfrak{D}_{3,\boldsymbol{\mu}} = - \mathrm{div} \boldsymbol{\mu}_3^{(1)} \nabla + \vec{a}_3 \cdot \nabla + M_3, 
$$
with a self-adjoint functional matrix $\boldsymbol{\mu}_3^{(1)}$, a scalar function $M_3$ and 
a functional vector $\vec{a}_3$. Then 
\begin{equation} \label{eq.d.S3.gen}
S_3 (d,\mathfrak{D}_{\boldsymbol{\mu}})= \left( \begin{array}{cccc} 
\mathfrak{D}_{3,\boldsymbol{\mu}} & {\rm div} & 0 & 0 \\ 
 - \nabla  & \mathfrak{D}_{2,\boldsymbol{\mu}}  & {\rm curl} & 0 \\ 
0 & {\rm curl}  & \mathfrak{D}_{1,\boldsymbol{\mu}}  &   \nabla \\ 
0 & 0 & - {\rm div} & \mathfrak{D}_{0,\boldsymbol{\mu}} \\
\end{array}
\right).
\end{equation}
Of course, we may double the dimension of the related matrices  $S_q (d,\mathfrak{D}_{\boldsymbol{\mu}})$ and 
change signs of entries outside the diagonal with the use of the imaginary unit. 

Let us interprete these Stokes' type operators within classical models of the 
Mathematical Physics. We consider three model examples only, because much more equations, 
that fit into our scheme,  
could be found in \cite{Mir15}, \cite{Mir18}, \cite{Mir20A}, \cite{Mir20B}. 

\begin{exmp} \label{ex.Maxwell}
We begin with the equations of electromagnetic field. Let $\mathfrak{c}$ stand for the speed of light.  
Then the Maxwell's type operator  
$\iota {\mathcal M}_{3} (A,\mathbf{b}\partial _t) $
gives us the classical Maxwell equations for electromagnetic field in a vacuum:
\begin{equation} \label{eq.d.S3.Maxw}
\iota \left( \begin{array}{cccc} 
 {\mathfrak c}^{-1}\partial _t  & {\rm div} & 0 & 0 \\ 
\nabla  &  {\mathfrak c}^{-1}\partial _t  &  {\rm curl} & 0 \\ 
0 & -{\rm curl}  &  {\mathfrak c}^{-1}\partial _t  &   \nabla \\ 
0 & 0 &  {\rm div} &  {\mathfrak c}^{-1}\partial _t \\
\end{array}
\right) 
\left( \begin{array}{cccc} 
0 \\ \vec{H} \\ \vec{E} \\ 0
\end{array}
\right) = \iota \left( \begin{array}{cccc} 
0 \\ 0 \\ -4\pi {\mathfrak c}^{-1} \vec{j}_e \\ 4\pi \rho_e
\end{array}
\right),
\end{equation}
where $\vec{E}$,  $\vec{H}$ represent the electric and magnetic field strengths, 
$\rho_e$ is the volume density of electric charge and $\vec{j}_e$ is the volume density of electric current  
(here the number $\iota$ is introduced for the related matrices to be self-adjoint).

Taking into account the magnetic charges and currents in the Dirac monopoles \cite{Dirac2}, \cite{Dirac3} and 
Schwinger dyons \cite{Schwinger} models, the equation  \eqref{eq.d.S3.Maxw} can be rewritten in more symmetric form
\begin{equation} \label{eq.d.S3.Maxw2}
\iota \left( \begin{array}{cccc} 
 {\mathfrak c}^{-1}\partial _t  & {\rm div} & 0 & 0 \\ 
\nabla  &  {\mathfrak c}^{-1}\partial _t  &  {\rm curl} & 0 \\ 
0 & -{\rm curl}  &  {\mathfrak c}^{-1}\partial _t  &   \nabla \\ 
0 & 0 &  {\rm div} &  {\mathfrak c}^{-1}\partial _t \\
\end{array}
\right) 
\left( \begin{array}{cccc} 
0 \\ \vec{H} \\ \vec{E} \\ 0
\end{array}
\right) = \iota \left( \begin{array}{cccc} 
4\pi \rho_m \\ -4\pi {\mathfrak c}^{-1} \vec{j}_m \\ -4\pi {\mathfrak c}^{-1} \vec{j}_e \\ 4\pi \rho_e
\end{array}
\right),
\end{equation}
where  $\rho_m$ is the volume density of magnetic charge and $\vec{j}_m$ is the volume density of 
magnetic current. 
Significantly, Lemma \ref{l.factor.t} immediately gives us the so-called 
wave equations  for the field's strength related to  
model \eqref{eq.d.S3.Maxw2}. Indeed, if we denote the d'Alembert operator as 
\[
\hat D = \left( {\frac{1}
{{\mathfrak{c}^2 }}\frac{{\partial ^2 }}
{{\partial t^2 }} - \Delta } \right),
\]
then we get
\begin{equation} \label{eq.d.S3.Maxw.wave}
 \left( \begin{array}{cccc} 
 \hat D & 0 & 0 & 0 \\ 
 0  &  \hat D &  0 & 0 \\ 
0 & 0  & \hat D  &   0 \\ 
0 & 0 &  0 & \hat D \\
\end{array}
\right) \!\!
\left( \begin{array}{cccc} 
0 \\ \vec{{H}} \\ \vec{{E}} \\ 0
\end{array}
\right) \!= \! \left( \begin{array}{cccc} 
f_1 \\ f_2  \\ f_3 \\ f_4
\end{array}
\right),
\end{equation}
where 

\[
\left( {\begin{array}{*{20}c}
   {f_1 }  \\
   {f_2 }  \\
   {f_3 }  \\
   {f_4 }  \\
 \end{array} } \right) = \left( {\begin{array}{*{20}c}
   {\mathfrak{c}^{ - 1} \partial _t }  & { - {\rm div}} & 0 & 0  \\
   { - \nabla } & {\mathfrak{c}^{ - 1} \partial _t }  & { - {\rm curl}} & 0  \\
   0 & {{\rm curl}} & {\mathfrak{c}^{ - 1} \partial _t }  & { - \nabla }  \\
   0 & 0 & { - {\rm div}} & {\mathfrak{c}^{ - 1} \partial _t }   \\
 \end{array} } \right)\left( {\begin{array}{*{20}c}
   {4\pi \rho _m }  \\
   { - 4\pi \mathfrak{c}^{ - 1} \vec j_m }  \\
   { - 4\pi \mathfrak{c}^{ - 1} \vec j_e }  \\
   {4\pi \rho _e }  \\
 \end{array} } \right).
\]
The matrix equation \eqref{eq.d.S3.Maxw.wave} is equivalent to the following system
\[
\left\{
\begin{gathered}
  \left( {\frac{1}
{{\mathfrak{c}^2 }}\frac{{\partial ^2 }}
{{\partial t^2 }} - \Delta } \right)\vec E =  - 4\pi \, \nabla \rho _e  - \frac{{4\pi }}
{{\mathfrak{c}^2 }}\frac{{\partial \vec j_e }}
{{\partial t}} - \frac{{4\pi }}
{\mathfrak{c}}\left[ { \nabla  \times \vec j_m } \right], \hfill \\
  \left( {\frac{1}
{{\mathfrak{c}^2 }}\frac{{\partial ^2 }}
{{\partial t^2 }} - \Delta } \right)\vec H =  - 4\pi \,\nabla \rho _m  - \frac{{4\pi }}
{{\mathfrak{c}^2 }}\frac{{\partial \vec j_m }}
{{\partial t}} + \frac{{4\pi }}
{\mathfrak{c}}\left[ { \nabla  \times \vec j_e } \right], \hfill \\
  \frac{{\partial \rho _m }}
{{\partial t}} + \left( { \nabla  \cdot \vec j_m } \right) = 0, \,\, 
  \frac{{\partial \rho _e }}
{{\partial t}} + \left( {\nabla  \cdot \vec j_e } \right) = 0, \hfill \\ 
\end{gathered} 
\right.
\]
where the last two relations are the laws of conservation of electric and magnetic charges.
\end{exmp}
\begin{exmp} \label{ex.Stokes} 
Let us discuss the matrix representation of the hydrodynamic equations. 
As we noted at the beginning of \S \ref{s.MS}, the Euler 
and the Navier-Stokes' equations 
for incompressible fluid fit perfectly  to this scheme with  Maxwell's type  operator 
operator ${\mathcal M}_1 (d, \partial _t)$ and Stokes' type operator 
$S_1 (d, \Delta, \partial _t)$ as principal linear parts, respectively, for instance, \cite{LaLi},
\begin{equation*} 
\left\{
\begin{gathered}
  \partial _t \vec v + (\vec v \cdot  \nabla )\vec v + \rho ^{ - 1} \nabla p = \vec{ f}, \hfill \\
  \mathrm{div}\,\vec v = 0, \hfill \\
 \end{gathered} 
\right.
\quad
\left\{
\begin{gathered}
 ( \partial _t -\mu \Delta) \vec v + (\vec v \cdot  \nabla )\vec v + \rho ^{ - 1} \nabla p = \vec{f}, \hfill \\
  \mathrm{div}\,\vec v = 0, \hfill \\
 \end{gathered} 
\right.
\end{equation*}
where $\vec{v}$ is a local flow velocity, $p$ is a pressure, $\rho$ is a fluid density, $\mu$ is viscosity and 
$\vec{f}$ is the vector of outer forces. 
Next, as is known \cite{LaLi}, 
the vortex-less free fluid is described by the following system of equations
\begin{equation} \label{eq.d.S3.Hyd.vortex-less_1}
\left\{
\begin{gathered}
  \partial _t \vec v + (\vec v \cdot  \nabla )\vec v + \rho ^{ - 1}  \nabla p = 0, \hfill \\
  \partial _t \rho  + (\vec v \cdot  \nabla )\rho  + \rho \, \mathrm{div}\,\vec v = 0, \hfill \\
  \mathrm{curl}\, \vec v = 0, \hfill \\ 
\end{gathered} 
\right.
\end{equation}
with the same entries as above.

Let us assume that the flow is isentropic (i.e., the entropy $s$ is a constant). Using the 
thermodynamic relation for enthalpy $h$ per unit mass 
$$
d\, h = Td{\kern 1pt} s + \rho ^{ - 1} d{\kern 1pt} p
$$ 
we can introduce new function $u$ according to the following relations
\[
d{\kern 1pt} u = \frac{1}
{\mathfrak{c}_s}d{\kern 1pt} h = \frac{1}
{{\mathfrak{c}_s\rho }}d{\kern 1pt} p = \frac{\mathfrak{c}_s}
{\rho }d\rho ,
\]
where $\mathfrak{c}_s$ is the speed of sound ($\mathfrak{c}_s^2  = \left( {\partial p/\partial \rho } \right)_s =const$). 
Then taking into account that the total time derivative is given as 
$$
\mathfrak{d}_t  = \partial _t  + (\vec v \cdot  \nabla ),
$$ 
we rewrite  the system \eqref{eq.d.S3.Hyd.vortex-less_1} in the following symmetric form
\[
\left\{
 \begin{gathered}
 {\mathfrak c} _s^{ - 1} \mathfrak{d}_t \vec v +  \nabla u = 0, \hfill \\
   {\mathfrak c}_s^{ - 1} \mathfrak{d}_t u + {\rm div}\vec v  = 0, \hfill \\
 {\rm curl}\;\vec v = 0, \hfill \\ 
\end{gathered} 
\right.
\]
or in the following matrix form
\begin{equation} \label{eq.d.S3.Hyd.vortex-less_2A}
\iota  \left( \begin{array}{cccc} 
 {\mathfrak c}^{-1}_s \mathfrak{d} _t  & {\rm div} & 0  \\ 
 \nabla  &  {\mathfrak c}^{-1}_s \mathfrak{d} _t  &  {\rm curl}  \\ 
0 & -{\rm curl}  &  {\mathfrak c}^{-1}_s \mathfrak{d}_t   \\ 
\end{array}
\right) 
\left( \begin{array}{cccc} 
u \\ \vec{v} \\ 0 
\end{array}\right )=\left( \begin{array}{cccc} 
 0\\ 0 \\ 0 
\end{array}\right ),
\end{equation}
where the principal linear part matches with the introduced above  Maxwell's type operator 
$\iota  {\mathcal M}_{2}(d,  {\mathfrak c}^{-1}_s \partial_t )$.

In order to interprete  Maxwell's type operator 
$\iota  {\mathcal M}_{3}(d,  {\mathfrak c}^{-1}_s \partial_t )$, we rewrite  
\eqref{eq.d.S3.Hyd.vortex-less_2A}: 
\begin{equation} \label{eq.d.S3.Hyd.vortex-less_2}
\iota  \left( \begin{array}{cccc} 
 {\mathfrak c}^{-1}_s \mathfrak{d} _t  & {\rm div} & 0 & 0 \\ 
 \nabla  &  {\mathfrak c}^{-1}_s \mathfrak{d} _t  &  {\rm curl} & 0 \\ 
0 & -{\rm curl}  &  {\mathfrak c}^{-1}_s \mathfrak{d}_t  &   \nabla \\ 
0 & 0 &  {\rm div} &  {\mathfrak c}^{-1}_s \mathfrak{d}_t \\
\end{array}
\right) 
\left( \begin{array}{cccc} 
u \\ \vec{v} \\0 \\  0
\end{array}\right )=\left( \begin{array}{cccc} 
 0\\ 0 \\ 0 \\ 0
\end{array}\right ).
\end{equation}
Next, to consider the vortex flow we introduce two new functions $\xi (\vec r,t)$ and ${\vec w}(\vec r,t)$ which 
describe the field of vortex tubes.  The value $\vec w$ is proportional to the angle vector of tube rotation, while 
$\xi$ characterizes the twisting of vortex tubes \cite{Mir20A}. Then the equation \eqref{eq.d.S3.Hyd.vortex-less_2} 
for vortex flow is written as
\begin{equation} \label{eq.d.S3.Hyd.vortex_1}
\left( \begin{array}{cccc} 
 {\mathfrak c}^{-1}_s \mathfrak{d}_t  & {\rm div} & 0 & 0 \\ 
 \nabla  &  {\mathfrak c}^{-1}_s \mathfrak{d}_t  &  {\rm curl} & 0 \\ 
0 & -{\rm curl}  &  {\mathfrak c}^{-1}_s \mathfrak{d}_t  &   \nabla \\ 
0 & 0 &  {\rm div} &  {\mathfrak c}^{-1}_s \mathfrak{d}_t \\
\end{array}
\right) 
\left( \begin{array}{cccc} 
u \\ \vec{v} \\ \vec w \\  \xi
\end{array}\right )=0.
\end{equation}
Finally, taking into the account the dissipation we obtain the related  equations
 of viscous vortex flow, including the Stokes' type operator 
$\iota S_{3,1}(d, {\mathfrak c}^{-1}_s\partial_t+\Delta_{\boldsymbol{\mu}} )$ with 
$\boldsymbol{\mu}^{(0)}_0 = \mu$, $\mu^{(0)}_1 =\boldsymbol{\mu}^{(1)}_1  = \mu I_3$, 
$\boldsymbol{\mu}^{(1)}_2 = \mu$ as the principal linear part:
\begin{equation*} 
\iota \left( \begin{array}{cccc} 
 {\mathfrak c}^{-1}_s (\mathfrak{d} _t  - \mu \Delta)& {\rm div} & 0 & 0 \\ 
 \nabla  &  {\mathfrak c}^{-1}_s (\mathfrak{d} _t -\mu \Delta) &  \mathrm{curl} & 0 \\ 
0 & -\mathrm{curl}  &  {\mathfrak c}^{-1}_s (\mathfrak{d} _t -\mu \Delta)  &   \nabla \\ 
0 & 0 &  {\rm div} &  {\mathfrak c}^{-1}_s (\mathfrak{d}_t -\mu \Delta)\\
\end{array}
\right) \!\!\!
\left( \begin{array}{cccc} 
\!u \!\\ \! \vec{v} \! \\ \! \vec{w} \! \\ \! \xi \!
\end{array}
\right) \!\!=0,
\end{equation*}
where the parameter $\mu$ represents kinematic viscosity. This matrix equation is equivalent to the following 
system, see \cite{Mir20A}:
$$
\left\{
\begin{gathered}
  \mathfrak{c}^{ - 1}_s \left( {\partial _t  + \left( {\vec v \cdot  \nabla } \right) - 
	\mu {\kern 1pt} \Delta } \right)\vec v + {\rm curl}\; \vec w +  \nabla u = 0, \hfill \\
  \mathfrak{c}^{ - 1}_s \left( {\partial _t  + \left( {\vec v \cdot \nabla } \right) - \mu {\kern 1pt} \Delta } \right)u + {\rm div}\;\vec v = 0, \hfill \\
   \mathfrak{c}^{ - 1}_s \left( {\partial _t  + \left( {\vec v \cdot  \nabla } \right) - \mu {\kern 1pt} \Delta } \right)\vec w - {\rm curl}\;\vec v +  \nabla \xi  = 0, \hfill \\
   \mathfrak{c}^{ - 1}_s \left( {\partial _t  + \left( {\vec v \cdot \nabla } \right) - \mu {\kern 1pt} \Delta } \right)\xi  + {\rm div}\;\vec w = 0. \hfill \\ 
\end{gathered} 
\right.
$$
\end{exmp}

\begin{exmp}
\label{ex.quatum.mass} 
Consider the  quadrupling of the imaginary de Rham complex over ${\mathbb R}^3$:
$ A_0 = \iota \, I_4 \otimes  \nabla$, $ A_1 =\iota  \, I_4 \otimes  \mathrm{curl}$, $A_2 = \iota \, I _4 \otimes 
 \mathrm{div}$. 
Then, for a real number $M$, we have   
$
A^*_0 = \iota \, I_4 \otimes  \mathrm{div}$, $ A^*_1 =  -\iota  \, I_4 \otimes  \mathrm{curl}$, 
$A^*_2 = \iota  \, I_4 \otimes \nabla$, $ (\iota \, M)^* = - \iota \, M$.

Set $\boldsymbol{\mu}^{(0)}_0 = 0$, $\boldsymbol{\mu}^{(0)}_1 = 0$, 
$\boldsymbol{\mu}^{(1)}_1 = 0$, 
$$
{\mathfrak D}_{0,\boldsymbol{\mu}} = \iota \, 
\left(
\begin{array}{lllll}
 0 & 0 & 0 & -M \\ 
0 & 0 & M & 0 \\
0 & -M & 0 & 0 \\
M & 0 & 0 & 0 \\
\end{array}
\right), \, 
{\mathfrak D}_{1,\boldsymbol{\mu}} = \iota \, 
\left(
\begin{array}{lllll}
 0 & -\mathrm{curl} & 0 & -M \\ 
\mathrm{curl} & 0 & M & 0 \\
0 & -M & 0 & \mathrm{curl} \\
M & 0 & -\mathrm{curl} & 0 \\
\end{array}
\right);
$$
in particular, the operators ${\mathfrak D}_{0,\boldsymbol{\mu}}$, 
${\mathfrak D}_{1,\boldsymbol{\mu}}$ are formally self-adjoint.
Then Stokes' operator $S_{0,\iota \mathfrak{c}^{-1}} (A,{\mathfrak D}_{\boldsymbol{\mu}},\partial_t )$ coincides with 
$\iota \mathfrak{c}^{-1} (\partial _t + {\mathfrak D}_{0,\boldsymbol{\mu}})$. 

Next, Stokes' system $S_{1,\iota \mathfrak{c}^{-1}} (A,{\mathcal L}_{\boldsymbol{\mu}})$ 
matches with the operator related to equations for the vector and scalar field’s
strengths in sedeonic field theory for the case of fields with non-zero mass of quantum 
$m_0$ \cite{Mir15},\cite{Mir20B}:
$$
 \left( \begin{array}{ccccc} 
\iota \mathfrak{c}^{-1} \partial_t +{\mathfrak D}_{1,\boldsymbol{\mu}} & \iota \, I_4 \otimes \nabla\\ 
\iota \, I_4 \otimes  \mathrm{div} & \iota \mathfrak{c}^{-1} \partial_t + {\mathfrak D}_{0,\boldsymbol{\mu}}  \\ 
\end{array}
\right).
$$
This leads us to the following matrix equation
\[
\left( {\begin{array}{*{20}c}
   \hat \partial _t  & 0 & 0 & { - M} & {\rm div} & 0 & 0 & 0  \\
   0 &  \hat \partial _t  & M & 0 & 0 & {\rm div} & 0 & 0  \\
   0 & { - M} &  \hat \partial _t& 0 & 0 & 0 & {\rm div} & 0  \\
   M & 0 & 0 &  \hat \partial  _t & 0 & 0 & 0 & {\rm div}  \\
   { \nabla } & 0 & 0 & 0 &  \hat \partial _t & {\rm - curl} & 0 & M  \\
   0 & { \nabla } & 0 & 0 & {\rm curl} &  \hat \partial _t & { - M} & 0  \\
   0 & 0 & {\nabla } & 0 & 0 & M &  \hat\partial _t & {\rm curl}  \\
   0 & 0 & 0 & { \nabla } & { - M} & 0 & {\rm - curl} &  \hat\partial _t  \\

 \end{array} } \right)\left( {\begin{array}{*{20}c}
   {g_1 }  \\
   {g_2 }  \\
   {g_3 }  \\
   {g_4 }  \\
   {\vec G_1 }  \\
   {\vec G_2 }  \\
   {\vec G_3 }  \\
   {\vec G_4 }  \\

 \end{array} } \right) = \left( {\begin{array}{*{20}c}
   {q _1 }  \\
   {q_2 }  \\
   {q_3 }  \\
   {q_4 }  \\
   {\vec J_1 }  \\
   {\vec J_2 }  \\
   {\vec J_3 }  \\
   {\vec J_4 }  \\

 \end{array} } \right),
\]
which is equivalent to the following system (see \cite{Mir15}):
\[
\left\{ 
\begin{gathered}
   {\mathfrak c}^{-1}\partial _t g_1  + {\rm div{\kern 1pt}} \vec G_1  - Mg_4  = 4\pi\rho _1 , \hfill \\
  {\mathfrak c}^{-1} \partial _t g_2  + {\rm div{\kern 1pt}} \vec G_2  + Mg_3  = 4\pi \rho _2 , \hfill \\
  {\mathfrak c}^{-1} \partial _t g_3  + {\rm div{\kern 1pt}} \vec G_3  - Mg_2  = 4\pi \rho _3 , \hfill \\
 {\mathfrak c}^{-1}  \partial _t g_4  + {\rm div{\kern 1pt}} \vec G_4  + Mg_1 \; = 4\pi \rho _4 , \hfill \\
  {\mathfrak c}^{-1} \partial _t \vec G_1  + \nabla g_1  - {\rm curl{\kern 1pt}} \vec G_2  + M\vec G_4  =  
- \frac{4\pi}{ {\mathfrak c}}\vec j_1 , \hfill \\
  {\mathfrak c}^{-1} \partial _t \vec G_2  +  \nabla g_2  + {\rm curl{\kern 1pt}} \vec G_1  - 
	M\vec G_3  =  - \frac{4\pi}{ {\mathfrak c}} \vec j_2 , \hfill \\
  {\mathfrak c}^{-1} \partial _t \vec G_3  +  \nabla g_3  + {\rm curl{\kern 1pt}} \vec G_4  + M\vec G_2  = 
 - \frac{4\pi}{ {\mathfrak c}} \vec j_3 , \hfill \\
  {\mathfrak c}^{-1} \partial _t \vec G_4  +  \nabla g_4  - {\rm curl{\kern 1pt}} \vec G_3  - M\vec G_1  =  
- \frac{4\pi}{ {\mathfrak c}} \vec j_4 . \hfill \\ 
\end{gathered} 
\right.
\]
Here $g_i$ and $\vec G_i$ are scalar and vector field strengths; 
 $q_i = 4\pi\rho _i$ (where $\rho _i$ 
are volume densities of charges); 
 $\vec J_i 
=- \frac{4\pi}{ {\mathfrak c}} \vec j_i $ (where $\vec j_i$ are volume densities of currents);
($i \in \left\{ {1,\;2,\;3,\;4} \right\}$);  
$\hat \partial _t =  {\mathfrak c}^{-1} \partial _t$; 
$M={m_0  {\mathfrak c}}/{\hbar}$, see \cite{Mir20B}.

The approach predicts two more Stokes' operators related 
to this mathematical model. One may conjecture that 
$\boldsymbol{\mu}^{(i)}_2 = 0$, $\boldsymbol{\mu}^{(1)}_3 = 0$ and 
 ${\mathfrak D}_{2,\boldsymbol{\mu}}$, ${\mathfrak D}_{3,\boldsymbol{\mu}}$ 
are given by
$$
{\mathfrak D}_{2,\boldsymbol{\mu}} \!\!= \!\! 
\left(
\begin{array}{lllll}
 \! 0 \!& \!\gamma_ 2 \mathrm{curl} \!& \! 0\! & \!\alpha_2 M \!\\ 
\!\overline \gamma_2  \mathrm{curl} \!& \!0 \!&\! \beta_2 M \!& \!0\!\\
\!0\! & \!\overline \beta_2 M \! & \!0\! & \!\delta_2  \mathrm{curl} \!\\
\!\overline \alpha_2 M \!& \!0 \!& \!\overline \delta_2  \mathrm{curl}\! &\! 0\! \\
\end{array}
\right) \!, \,   
{\mathfrak D}_{3,\boldsymbol{\mu}} \!\!=\!\! 
\left(
\begin{array}{lllll}
 \!0 \!&  \! 0  \! &  \!0  \!&  \!\alpha_3 M  \! \\ 
 \! 0  \! &  \! 0  \! &  \! \beta_3 M  \! &  \! 0  \! \\
 \!0  \! &  \!\overline \beta_3 M  \! &  \! 0  \! &  \! 0  \! \\
 \!\overline  \alpha_3 M  \!& \! 0  \! &  \! 0  \! &  \! 0  \! \\
\end{array}
\right), \,
$$
respectively, with complex numbers $\alpha_j, \beta_j$, $\gamma_j$, $\delta_j$ (highly likely, $\pm \iota$)
and then 
$$
S_{2,\mathfrak{c}^{-1}\iota} (A,\partial_t+{\mathfrak D}_{\boldsymbol{\mu}}) = 
 \left( \begin{array}{ccccc} 
\iota \mathfrak{c}^{-1}\partial_t + {\mathfrak D}_{2,\boldsymbol{\mu}} &  \iota \, I_4 \otimes  \mathrm{curl} & 0 \\ 
-\iota  \, I_4 \otimes  \mathrm{curl} & \iota \partial_t + {\mathfrak D}_{1,\boldsymbol{\mu}} &  
\iota  \, I_4 \otimes \nabla\\ 
0 & \iota \, I_4 \otimes  \mathrm{div} & \iota \partial_t + {\mathfrak D}_{0,\boldsymbol{\mu}}  \\ 
\end{array}
\right),
$$
$$
S_{3,\iota} (A,{\mathcal L}_{\boldsymbol{\mu}}) \!= \!
 \iota \left( \begin{array}{ccccc} 
 \partial_t +{\mathfrak D}_{3,\boldsymbol{\mu}} &    I_4 \otimes  \mathrm{div} & 0 & 0 \\ 
  I_4 \otimes \nabla &  \mathfrak{c}^{-1} \partial_t + {\mathfrak D}_{2,\boldsymbol{\mu}} &  I_4 \otimes  \mathrm{curl} & 0 \\ 
0 & -  I_4 \otimes  \mathrm{curl} & \mathfrak{c}^{-1} \partial_t + {\mathfrak D}_{1,\boldsymbol{\mu}} &  
 I_4 \otimes \nabla\\ 
0 & 0 &  I_4 \otimes  \mathrm{div} &  \mathfrak{c}^{-1} \partial_t + {\mathfrak D}_{0,\boldsymbol{\mu}}  \\ 
\end{array}
\right).
$$
\end{exmp}

\section{Some typical elliptic differential complexes} 
\label{s.other}

Consider some other typical examples of elliptic complexes. 

\begin{exmp} \label{ex.Koszul}
A large part of complexes represents the so-called Koszul complexes. 
Namely, let $A_0$ be a column of scalar differential operators $(Q_1, \dots Q_N)^T$ over 
an open set $X\subset {\mathbb R}^n$, $1\leq N\leq n$, satisfying the following 
commutation assumptions: 
\begin{equation}\label{eq.Koszul}  
Q_i Q_j = Q_j Q_i \mbox{ for all } 1\leq i< j \leq N.
\end{equation}
Setting $E_q = X \times {\mathbb C}^{k_q}$ with $k_q =\Big(\begin{array}{ll} N \\ q
\end{array}\Big)$ we may define differential operators 
$$
A_q = \sum_{\# I=q} Q_i (x)  dy_i \wedge dy_I , \, x \in X\subset {\mathbb R}^n, 
$$
where $y= (y_1, \dots, y_N)$ are coordinates in ${\mathbb R}^N$. Again $A_{q+1} \circ A_q =0$ 
and hence we obtain a differential complex $\{ A_q, E_q\}_{q=0}^N$, 
see \cite[\S 1.2.8]{Tark35}, that is usually called Koszul complex associated with the set 
$(Q_1, \dots Q_N)$. 

According to \cite[Proposition 1.2.51]{Tark35}
this complex is elliptic if and only if the principal symbol of the operator 
$A_0 =  \sum_{i=1}^N Q_i (x)  dy_i $ 
is injective; of course we may interpret $A_j$ as matrix differential operators:
$$
A_0 = \left(\begin{array}{llll} Q_1 \\ \dots \\ Q_N \\
\end{array}\right) , \, \dots , \, A_{N-1} = \big(\begin{array}{llll} Q_1 ,\dots  , Q_N \\
\end{array}\big)
$$
In particular, for operators with constant coefficients we have $Q_j^* Q_i = Q_i Q_j^*$ and 
$$
\Delta_q = \Big(\sum_{i=1}^N Q_i^* Q_i \Big) I_{k_q}.
$$
If the operators $Q_i$ has the same order then the Laplacians $\Delta_q$ are strongly elliptic. 

Unfortunately the Koszul complexes are not always compatibility complexes.
For operators with constant coefficients one may use \cite[Proposition 1.2.52]{Tark35}
giving a simple sufficient condition providing the compatibility property:
the dimension of the algebraic variety 
$$
{\mathcal N} (A_0) = \{ z\in {\mathbb C}^n: Q_1 (z) = \dots =Q_N (z) =0\}
$$
is no more than $(n-N)$. Again, for $N=3$ with the use
of the classical algebraic constructions we obtain: 
$$
A_0 h = ( Q_1 h , Q_2 h , Q_3 h )^T, \,  
A_1 \vec{u} =A_0 \times  \vec{u}, \, A_2 \vec{v} = A_0 \cdot \vec{v}  ,
$$
for (vector-)functions $h(x)$, $\vec{u}(x)$, $\vec{v}(x)$  of $n$ variables $x = (x_1, \dots x_n)$ with $n\geq 3$.

Of course, the initial operator $A_0$ may be non-homogeneous.
\end{exmp}

\begin{exmp} \label{ex.deRham.p}
Let $X={\mathbb R}^n$ and $p \in \mathbb N$, $p\geq 2$. 
For a $q$-differential form $u$ 
we set 
$$
A^{(p)}_q u(x) = \sum_{i=1}^n
\sum_{\# I=q} \partial^p_i  
u_I(x) dx_i \wedge dx_I  .
$$
Using \eqref{eq.wedge} we easily conclude that  
\begin{equation} \label{eq.deRham.p.base.rel} 
A^{(p)}_{q+1} \circ A^{(p)}_q =0
\end{equation}
 and then we obtain a Koszul complex $\{ A^{(p)}_q, \Lambda^q\}$.
In this case $N=n$ and the algebraic variety 
${\mathcal N} (A^{(p)}_0)$ is trivial, i.e. $\{ A^{(p)}_q, \Lambda^q\}$ is a 
compatibility complex for the operator $A^{(p)}_0 = \left(
\begin{array}{ccc}
\partial_1^p 
\\
\dots \\
\partial_n^p 
\end{array}
\right)$, see Example \ref{ex.Koszul} above. 
Of course, in addition to  \eqref{eq.deRham.p.base.rel},  we also have 
$
\Delta^{(p)}_q = (A^{(p)}_q) ^{*} A^{(p)}_q + A^{(p)}_{q-1} (A^{(p)}_{q-1}) ^{*} = 
(-1)^p \Big( \sum_{j=1}^n \partial^{2p}_j 
\Big) I_{k_q}$, $ 0\leq q\leq n$. 
Again we still may define this compatibility complex in ${\mathbb R}^3$ with the use
of the classical algebraic constructions: 
$A_1^{(p)} u =A_0^{(p)} \times  u$, $ A^{(p)}_2 v = A_0^{(p)} \cdot v  $. 
\end{exmp}

\begin{exmp} \label{ex.mixed}
Consider the differential operator  
$$
A = \left( \begin{array}{cccc}
0 & -\partial_3
\\
\partial_3
 & 0 \\ 
-\partial_2 
& \partial_ 1  
 \\
- \partial_1  
& -\partial_2 
\end{array}
\right)
$$
in ${\mathbb R}^3$, that is closely related to the de Rham complex on the plane $Ox_1x_2$:
$$
d_0 = \nabla_2 = \left( 
\begin{array}{llll}
\partial_1  \\
 \partial_2 \\
\end{array}
\right)
, \, d_1 = \mathrm{curl}_2 = (-\partial_2 ,  \partial_ 1 ), \,  
d^*_0 = -\mathrm{div}_2  =-(\partial_1 ,  \partial_ 2 ).
$$
Alledgedly, the related system of equations
\begin{equation}\label{eq.planar.flow}
\left\{ 
\begin{array}{llll}
\partial_3 v_1 = 0, \\  
\partial_3 v_2 = 0, \\
\mathrm{curl}_2 \vec v = -\partial_2 v_1  +\partial_ 1 v_2 = f_1, \\
-\mathrm{div}_2 \vec v =  -\partial_1v_1 -  \partial_2 v_2 =f_2, \\
\end{array}
\right.
\end{equation}
was pointed out by L. Euler for the description of the velocity $\vec v = (v_1, v_2)$ of a 
plane-parallel flow on layers $\{ x_3 = const\}$ with given 'plane rotation' $f_1$ and 'plane source' $f_2$ in the 
case where the flow does not depend on the layer. Taking a complex valued functions 
$ \tilde v (x_1,x_2) = -\iota (v_1 - 
\iota v _2)$, $\tilde f (x_1,x_2) =  f_1 - 
\iota f _2$ one easily reduces the last two equations in \eqref{eq.planar.flow} to the non-
homogeneous Cauchy-Riemann system $\overline \partial 
\tilde v = \tilde f$ on the plane $Ox_1x_2$.  

Clearly, for the differential 
operator 
$$
B = \left( \begin{array}{ccccc}
\partial_2  &  -\partial_1 & 0 & -\partial_3  \\
 \partial_1  &  \partial_2  & \partial_3 &  0  \\
\end{array}
\right)
$$ 
we have $B \circ A\equiv 0$. Passing to the polynomial matrices, we see that 
if a differential operator $\tilde B$ satisfies $\tilde B \circ A\equiv 0$ then
for $\tilde B (\zeta) =  (b_1 (\zeta), b_2 (\zeta), b_3 (\zeta), b_4(\zeta))$ 
we have 
\begin{equation} \label{eq.CR.compat}
\zeta_3 b_2 - \zeta_2 b_3 - \zeta_1 b_4  =0, \, 
-\zeta_3 b_1+\zeta_1 b_3 - \zeta_2 b_4  =0, 
\end{equation}
and hence
$$
|\zeta|^2 b_3  = \zeta_1\zeta_3 b_1 + \zeta_2 \zeta_3 b_2 , \, 
|\zeta|^2 b_4  =- \zeta_2\zeta_3 b_1 + \zeta_1 \zeta_3 b_2. 
$$
In particular, this means that there are polynomials $c_1 (\zeta)$, $c_2 (\zeta)$
such that 
$$
b_3 (\zeta)= \zeta_3 c_1 (\zeta) ,\, b_4 (\zeta)= \zeta_3 c_2 (\zeta) 
$$ 
and then, taking into account \eqref{eq.CR.compat},  
$$
\tilde B (\zeta) =  ( \zeta_1  c_1 - \zeta_2  c_2, \zeta_1  c_2 + 
\zeta_1  c_2, \zeta_3 c_1 , \zeta_3 c_2  ) = \big (-c_2, \, c_1  \big) B(\zeta).
$$
Thus, $\{ A_0=A, A_1=B\}$ is a compatibility complex. Moreover,  
as the related Laplacians have the following form: 
$\Delta_0 = 
- \Delta I_2$, $\Delta_1= 
- \Delta I_4$, $ \Delta_2 = 
- \Delta I_2$, 
where $\Delta$ is the usual Laplace operator  in ${\mathbb R}^3$, 
we conclude that this complex is elliptic. 

For system \eqref{eq.planar.flow} the compatibility conditions induced by the operator $B$ just mean that  the data 
$f_1$, $f_2$ do not depend on the variable $x_3$:
$$
\partial_3 f_1 = \partial_3 f_2  = 0.
$$
Situation becomes more complicated if we consider the system
\begin{equation}\label{eq.planar.vel}
\left\{ 
\begin{array}{llll}
\partial_3 v_1 = g_1  (x_1,x_2, x_3),  \\
\partial_3 v_2 = g_2  (x_1,x_2, x_3), \\
\mathrm{curl}_2 \vec v= \partial_1 v_2  -\partial_ 2 v_1 = f_1 (x_1,x_2, x_3), \\
-\mathrm{div}_2 \vec v = -\partial_1v_1 - \partial_2 v_2 =f_2  (x_1,x_2, x_3), \\
\end{array}
\right.
\end{equation}
i.e. we are looking for 'planar' velocity $\vec v = (v_1, v_2)$ of the 
flow on each layer $\{ x_3 = const\}$ with given 'plane rotation' $f_1$ and 'plane source' $f_2$ in the case where 
the flow  {\it depends} on the layers. Then the compatibility conditions are the following:
\begin{equation*} 
\left\{ 
\begin{array}{llll}
\mathrm{curl}_2 \, \vec g = \partial_ 1 g_2- \partial_2 g_1   = 
\partial_ 3 f_1 (x_1,x_2, x_3), \\
- \mathrm{div}_2 \, \vec g =-\partial_1g_1 -  \partial_2 g_2 =\partial _3 f_2  (x_1,x_2, x_3). \\
\end{array}
\right.
\end{equation*}
The related linear Maxwell' type operators ${\mathcal M}_j (A,\partial _t )$ are the following:
$$
{\mathcal M}_1 (A,\partial _t ) = \left( 
\begin{array}{cccccc} 
\partial_t  & 0 & 0 & 0 & 0 & -\partial_3 \\ 
0  &  \partial_t & 0 & 0& \partial_3 & 0 \\ 
0  &  0 & \partial_t  & 0& -\partial_2 & \partial_1 \\ 
0  &  0 & 0&  \partial_t  & -\partial_1 & -\partial_2 \\ 
0 & -\partial_3 & \partial_2 & \partial_1  & \partial_t  & 0 \\
\partial_3 &  0 & -\partial_1  & \partial_2  & 0 &\partial_t   \\
\end{array}
\right) ,
$$
$$
{\mathcal M}_2 (A,\partial _t ) = \left( 
\begin{array}{cccccccc} 
\partial_t & 0 & \partial_2  & -\partial_1 & 0  &  -\partial_3 & 0 &  0   \\
0 & \partial _t  & \partial_1 & \partial_2 &  \partial_3 & 0 &  0 & 0  \\
-\partial_2 & -\partial_1 &\partial_t  & 0 & 0 & 0 & 0 & -\partial_3 \\ 
\partial_1& -\partial_2 &0  &  \partial_t & 0 & 0& \partial_3 & 0 \\ 
0& -\partial_3  &0 &  0 & \partial_t  & 0& -\partial_2 & \partial_1 \\ 
\partial_3 & 0 &0  &  0 & 0&  \partial_t  & -\partial_1 & -\partial_2 \\ 
0& 0 &0 & -\partial_3 & \partial_2 & \partial_1  & \partial_t  & 0 \\
0& 0 &\partial_3 &  0 & -\partial_1  & \partial_2  & 0 &\partial_t   \\
\end{array}
\right) .
$$
Now, taking into the account the following relation between the two-dimensional and 
the three-dimensional rotation operators, 
$ 
\mathrm{curl}_2 (v_1 , v_2) = \mathrm{curl}_3 (v_1 , v_2,0), 
$  
we see that the related Maxwells' system  ${\mathcal M}_1 (A,\partial _t )$ 
is  reduced to  system ${\mathcal M}_2 (d,\partial _t )$ 
for the de Rham complex with $\vec v = (v_1,v_2,v_3)$ truncated to $(v_1,v_2,0)$ that 
 corresponds to one line missing in \eqref{eq.d.S3.Hyd.vortex-less_2A}: 
$$
-\partial_2 w_1 + \partial_1 w_2 + \partial_3  u + \mathfrak{d}_t v_3  =0.
$$
The similar fact is valid for Stokes' operators 
 $S_1 (A,{\mathfrak D}_{\boldsymbol{\mu}},\partial _t )$ and 
$S_2 (d,{\mathfrak D}_{\boldsymbol{\mu}},\partial _t )$
with the  missing  line in the corresponding non-linear perturbation of 
$S_2 (d,{\mathfrak D}_{\boldsymbol{\mu}},\partial _t )$: 
$$
-\partial_2 w_1 + \partial_1 w_2 + \partial_3  u +(\mathfrak{d}_t -\mu \Delta) v_3  =0. 
$$
Significantly, ${\mathcal M}_2 (A,\partial _t )$ coincides with ${\mathcal M}_3 (d,\partial _t )$  
up to the order of lines and columns, and similarly for  
$S_2 (A,{\mathfrak D}_{\boldsymbol{\mu}},\partial _t )$ and  $S_3 (d,{\mathfrak D}_{\boldsymbol{\mu}},\partial _t )$, 
i.e. the missing component $v_3$ and the missing lines are restored automatically on the last step. 
\end{exmp} 

\begin{exmp} \label{ex.?2}
Let $X={\mathbb R}^2$ and 
$$
A =\left( \begin{array}{cccc}
\partial_1 
& 0 \\
\partial_2
& \partial_1 
\\
0 & \partial_2
\\
\end{array}
\right). 
$$
Then its principal symbol is injective because
$$
A^*A = -\left( \begin{array}{cccc}
 \Delta & \partial_1 \partial_2 
 \\
\partial_1 \partial_2 
& \Delta  \\
\end{array}
\right),  \det(\sigma(A^*A )(\zeta)) = |\zeta|^4 - \zeta_1^2 \zeta_2^2 >0
\mbox{ for all } \zeta \in {\mathbb R}^2\setminus \{0\}.
$$
Clearly, for the polynomial vector $B(\zeta)  = (\zeta_2^2 , - \zeta_1\zeta_2, \zeta_1^2)$ 
we have $B(\zeta) A(\zeta) \equiv 0$. If a vector 
$\tilde B(\zeta)= (b_1(\zeta), b_2(\zeta), b_3(\zeta))$ satisfies 
$\tilde B(\zeta) A(\zeta) \equiv 0$ then 
$$
b_1 (\zeta)\zeta_1 + b_2 (\zeta)\zeta_2 =  b_2 (\zeta)\zeta_1 + b_3 (\zeta)\zeta_2=0.
$$
Hence $b_1 \zeta_1^2 = b_3 \zeta_2^2$, and for the polynomial $p (\zeta) = b_3/\zeta_1^2$ we have 
$\tilde B(\zeta) = p(\zeta)B(\zeta)$,  
i.e. the differential operator $B=(\partial_2^2 , - \partial_1\partial_2, \partial_1^2)$
is a compatibility operator for $A$.  As the mapping $B(\zeta): {\mathcal P}^3 \to {\mathcal P}$ is 
surjective for $\zeta \in {\mathbb R}^n \setminus \{0\}$, we see that any polynomial $C(\zeta)$, 
satisfying $C(\zeta) B(\zeta) =0$, is identically zero. Thus operators $A$ and $B$ form a 
compatibility differential complex.

Moreover, if $w = (w_1 , w_2, w_3 )^T$ is a complex vector, satisfying
$$
B(\zeta) w = \zeta_2^2 w_1   - \zeta_1\zeta_2 w_2  + w_3 \zeta_1^2 = 0
$$ 
then there is a complex vector $v=(v_1,v_2)^T $ such that $w=A(\zeta)v $ if $\zeta \in {\mathbb R}^2 \setminus 
\{0\}$: 
$$
w_1=0, v_1 = \frac{w_2}{\zeta_2}, v_2 = \frac{w_3}{\zeta_2} \mbox{ if } \zeta_1=0, \zeta_2 \ne 0,  
$$
$$
w_3=0, v_1 = \frac{w_1}{\zeta_1}, v_2 = \frac{w_2}{\zeta_1} \mbox{ if } \zeta_1\ne 0, \zeta_2 = 0,  
$$
$$
w_2=\frac{\zeta_2^2 w_1   + w_3 \zeta_1^2}{ \zeta_1\zeta_2}, v_1 = \frac{w_1}{\zeta_1}, v_2 = \frac{w_3}{\zeta_2} \mbox{ if } \zeta_1\ne 0, \zeta_2 \ne 0.  
$$
Thus, the  range of the mapping  
$A (\zeta): {\mathbb C} ^2\to {\mathbb C}^3$ coincides with the kernel of the mapping  
$B (\zeta): {\mathbb C} ^3\to {\mathbb C}^1$ if $\zeta \in {\mathbb R}^2 \setminus 
\{0\}$. As  the mapping $A (\zeta)$
is injective and the mapping $B (\zeta)$
is surjective, we conclude that the related complex is elliptic. 

However, the operators in the complex have different orders and, unfortunately, we can not 
decrease the order of $B$. Of course, 
$\Delta_{2} = \Delta^2 - \partial^2_1\partial^2_2$ 
is a strongly elliptic operator, 
but taking $\boldsymbol{\mu}_0^{(0)} = - \Delta I_2$,
$\boldsymbol{\mu}_1^{(1)} = - \Delta I_3$
we obtain 
$$
\Delta_{0,\boldsymbol{\mu}_0} \!= \!\left( \begin{array}{cccc}
 \Delta ^2 & \partial_1 \Delta \partial_2 
 \\
\partial_1 \Delta \partial_2
& \Delta^2  \\
\end{array}
\right),  
\Delta_{1,\boldsymbol{\mu}_1} \!= \!\left( \begin{array}{cccc}
\Delta ^2 - \partial_1 ^2 \partial_2 ^2  & \partial_1^3 \partial_2 & \partial_1 ^2 \partial_2 ^2  \\
\partial_1 ^3 \partial_2
& \Delta ^2 + \partial_1 ^2 \partial_2 ^2 & \partial_1 \partial_2^3 \\
\partial_1 ^2 \partial_2^2 & \partial_1 \partial_2^3 
& \Delta ^2 - \partial_1 ^2 \partial_2 ^2 \\
\end{array}
\right),
$$
strongly elliptic non-negative self-adjoint operators.
\end{exmp} 

\begin{exmp} 
Consider the compatibility complex for the multidimensional Cau\-chy-Riemann 
operator $\overline \partial$
in ${\mathbb C}^n \cong {\mathbb R}^{2n}$, $n>1$, i.e. for $n$-vector column with the components 
$\frac{\partial }{\partial \overline z_j}$, $1\leq j \leq n$, where, as usual,
$z_j = x_{2j-1} + \iota x_{2j}$, $ \overline z_j = x_{2j-1} - \iota x_{2j} 
$, 
$$
\frac{\partial }{\partial z_j} = \frac{1}{2}\left( \frac{\partial }{\partial x_{2j-1}} -  \iota
\frac{\partial }{\partial x_{2j}} \right), \,  
\frac{\partial }{\partial \overline z_j} = \frac{1}{2}\left( \frac{\partial }{\partial x_{2j-1}} +  \iota
\frac{\partial }{\partial x_{2j}} \right), \, 1\leq j \leq n.
$$
Then the Dolbeault complex is the elliptic compatibility Koszul complex related to the column $\overline \partial$, see \cite[\S 1.2.7]{Tark35} or Example \ref{ex.Koszul} above. As  
$
\overline \partial^* = - \left( \begin{array}{llll}
\frac{\partial }{\partial  z_1}, & 
\dots , & \frac{\partial }{\partial  z_n}
\end{array}
\right) 
$, then  we easily calculate the related Laplacians $\Delta_q = (-1/4) \Delta I_{k_q}$, 
$0\leq q \leq n$, where $\Delta$ is the 
usual Laplace operator in ${\mathbb R}^{2n}$ and $k_q = \Big(\begin{array}{ll} n \\ q 
\end{array}\Big)$.

The simplest related Stokes' type operators arising in ${\mathbb C}^2 \cong 
{\mathbb R}^4$ with the coordinates 
$(z_1,z_2) =(x_1 + \iota x_2 , x_3 + \iota x_4 )\cong (x_1,y_1,x_2,y_2)
$ can be written as follows:
$$
S_1 (\overline \partial,\partial _t+\Delta_{\boldsymbol{\mu}} ) = 
\left(
\begin{array}{lllll}
\partial_t -\mu \Delta & 0 & 0 & 0 & \mathrm{curl}_{x_1,x_2}  \\
0 & \partial_t -\mu \Delta & 0 & 0 &  \mathrm{div}_{x_1,x_2} \\
0 & 0 & \partial_t -\mu \Delta & 0 &  \mathrm{curl}_{x_3,x_4}  \\
0 & 0 &  0 & \partial_t -\mu \Delta &  \mathrm{div}_{x_3,x_4} \\
\mathrm{curl}^*_{x_1,x_2}  & -\nabla_{x_1,x_2} &  \mathrm{curl}^*_{x_3,x_4}  
 & -\nabla_{x_3,x_4} & (\partial_t -\mu \Delta)\, I_2 \\
\end{array}
\right),
$$
where $\mu$ is a positive real number and 
$$
\mathrm{curl}_{x_{2j-1}, x_{2j}} = (-\partial_{2j} ,  \partial_{2j-1} ), \,  
\mathrm{div}_{x_{2j-1}, x_{2j}}  =(\partial_{2j-1} ,  \partial_{2j} ).
$$
\end{exmp}

{\sc Acknowledgments.}
The first author was supported 
by the Krasnoyarsk Mathematical Center and financed by the Ministry of Science and Higher 
Education of the Russian Federation (Agreement No. 075-02-2024-1429).

\end{document}